\documentclass[12pt]{article}
\usepackage[utf8]{inputenc}

\title{The termination of Nielsen transformations applied to word equations with length constraints}
\author{Benjamin Przybocki \and Clark Barrett}
\date{\today}

\usepackage{graphicx}
\usepackage{amsmath,amsthm,amssymb,mathrsfs}
\usepackage{fullpage}
\usepackage{mathtools}
\usepackage{hyperref}
\usepackage{comment}
\usepackage[english]{babel}
\usepackage{float}
\usepackage{enumitem}
\usepackage{tikz}
\usetikzlibrary{shapes.geometric, arrows}
\tikzstyle{arrow} = [very thick,->,>=latex]
\tikzset{vertex/.style = {shape=circle,draw,minimum size=1.5em}}
\tikzset{edge/.style = {->,> = latex'}}

\hypersetup{
colorlinks = true,
citecolor = blue,
linkcolor = black
}

\theoremstyle{definition}
\newtheorem{definition}{Definition}[section]

\theoremstyle{plain}
\newtheorem{theorem}[definition]{Theorem}

\newtheorem{proposition}[definition]{Proposition}
\newtheorem{lemma}[definition]{Lemma}

\newtheorem{example}[definition]{Example}

\newtheorem{conjecture}[definition]{Conjecture}

\theoremstyle{remark}

\newcommand{\ang}[1]{\langle #1 \rangle}

\allowdisplaybreaks

\begin{document}

\maketitle

\begin{abstract}
    Nielsen transformations form the basis of a simple and widely used procedure for solving word equations. We make progress on the problem of determining when this procedure terminates in the presence of length constraints. To do this, we introduce \emph{extended word equations}, a mathematical model of a word equation with partial information about length constraints. We then define \emph{extended Nielsen transformations}, which adapt Nielsen transformations to the setting of extended word equations. We provide a partial characterization of when repeatedly applying extended Nielsen transformations to an extended word equation is guaranteed to terminate.
\end{abstract}

\tableofcontents

\section{Introduction}

This paper concerns word equations and their solutions, which are defined as follows. Given a set $X$, let $X^*$ denote the free monoid generated by $X$, where we use juxtaposition to represent the binary operation and $\varepsilon$ to denote the identity element. We fix an \emph{alphabet} $\Gamma$, whose elements are called \emph{letters}, and a set of \emph{variables} $\mathcal{X}$. A \emph{word equation} is a pair $(U_1,U_2)$, where $U_1,U_2 \in (\mathcal{X} \cup \Gamma)^*$. A \emph{solution} is a homomorphism $S : (\mathcal{X} \cup \Gamma)^* \to \Gamma^*$ such that $S(U_1) = S(U_2)$ and $S(\mathtt{x}) = \mathtt{x}$ for all $\mathtt{x} \in \Gamma$.

\begin{example}
    Let $\Gamma = \{\mathtt{a}, \mathtt{b}, \mathtt{c}\}$ and $\mathcal{X} = \{X, Y\}$. Then, $(\mathtt{a}X\mathtt{ba}, \mathtt{ac}Y\mathtt{a})$ is a word equation. One solution is induced by $S(X) = \mathtt{caca}$ and $S(Y) = \mathtt{acab}$, since $S(\mathtt{a}X\mathtt{ba}) = S(\mathtt{ac}Y\mathtt{a}) = \mathtt{acacaba}$.
\end{example}

Understanding the solutions to word equations is not only a theoretically fruitful mathematical problem but also a practical one for computer science given the ubiquity of data structures with concatenation, such as strings and sequences (see \cite{amadini2021, day2022}). Automatically finding solutions for word equations plays an important part in automated reasoning algorithms, especially when using satisfiability modulo theories (SMT) solvers to reason about strings and sequences. Makanin~\cite{makanin1977} proved that the problem of whether a word equation has a solution is decidable, but his algorithm is remarkably complicated. Now, there are simpler algorithms available, such as the one given by Je\.{z}~\cite{jez2016, jez2020}.

In real-world applications, it is often necessary to solve word equations in the presence of other constraints, such as length constraints. Since the decidability of word equations with length constraints is an open problem \cite{decidability}, implementations of algorithms for reasoning about strings and sequences are based on incomplete procedures. Most SMT solvers for strings and sequences (e.g., \cite{z3str3,rulebased,dpllstrings,sequences}) use some variant of \emph{Nielsen transformations} to reason about word equations. This procedure, which is described below, is very simple but non-terminating in general. Understanding cases where it does terminate in the presence of length constraints has hitherto received limited attention, which we hope to rectify in this paper.\footnote{In \cite{dpllstrings}, the authors say they ``would like to identify fragments where our calculus is terminating'', which essentially boils down to understanding when the repeated application of Nielsen transformations terminates.}

One of the problems that arises in this context is how to mathematically model word equations when partial information is known about the relative lengths of variables. The simplest solution would be to forget about length constraints and just work with raw word equations. But this is unsatisfactory because the length constraints are often essential for termination, since they limit which transformations are applicable. Our solution is to define extended word equations, which include partial information about the relative lengths of variables in the word equations. We then define extended Nielsen transformations, which are Nielsen transformations applied to extended word equations.

The primary purpose of this paper is to determine when the procedure of repeatedly applying extended Nielsen transformations to an extended word equation is guaranteed to terminate. Theorem~\ref{thm-sufficient} gives a sufficient condition for an extended word equation to be terminating. This condition is not quite necessary, but Theorem~\ref{thm-necessary} proves that the condition is necessary for a large class of extended word equations. We additionally prove some miscellaneous propositions about extended word equations in Section~\ref{sec-misc}.

Besides our definitions of extended word equations and extended Nielsen transformations, our primary conceptual innovation is the definition of the cut graph of an extended word equation (see Section~\ref{sec-cut-graph}), which encodes information about dependencies between variable instances in an extended word equation. The proofs of our results will illustrate the importance of the cut graph to understanding extended word equations. We hope that this concept will help shine further light on problems regarding word equations with length constraints.

\section{Preliminaries} \label{sec-prelim}

For our purposes, it is simpler to replace each letter in a word equation by a variable. This yields a word equation of the form $(U_1,U_2)$, where $U_1,U_2 \in \mathcal{X}^*$. Hereafter, by ``word equation'', we will always mean a word equation in which any letters have been replaced by variables. Of course, this changes the solutions of the word equation (in particular, $S(X) = \varepsilon$ for all $X \in \mathcal{X}$ is always a solution), but it simplifies the description of some of the procedures that follow.

We begin by recalling the standard definition of Nielsen transformations as a way of motivating the definition of extended Nielsen transformations. First, we need the following bit of notation: if $U_i = U_{i,1} U_{i,2} \cdots U_{i,k} \in \mathcal{X}^*$, where $i\in\{1,2\}$ and $k \ge 1$, then let $U_i^+ = U_{i,2} \cdots U_{i,k}$. Now, given a word equation $(U_1,U_2)$, write
\begin{align*}
    U_1 &= U_{1,1} U_{1,2} \cdots U_{1,m} \quad\text{and} \\
    U_2 &= U_{2,1} U_{2,2} \cdots U_{2,n},
\end{align*}
where $U_{i,j} \in \mathcal{X}$ for all $i,j$. If $m,n \neq 0$, then a \emph{Nielsen transformation} of $(U_1,U_2)$ transforms the word equation into the word equation $(U'_1,U'_2)$, where $U'_1$ and $U'_2$ are given by one of the following three cases:
\begin{enumerate}[label=\Roman*.]
    \item $U'_1 = T(U_1)^+$ and $U'_2 = T(U_2)^+$, where $T$ is the endomorphism on $\mathcal{X}^*$ given by $T(U_{2,1}) = U_{1,1}$ and the identity function on other elements of $\mathcal{X}$.
    \item $U'_1 = T(U_1)^+$ and $U'_2 = T(U_2)^+$, where $T$ is the endomorphism on $\mathcal{X}^*$ given by $T(U_{1,1}) = U_{2,1} U_{1,1}$ and the identity function on other elements of $\mathcal{X}$.
    \item $U'_1 = T(U_1)^+$ and $U'_2 = T(U_2)^+$, where $T$ is the endomorphism on $\mathcal{X}^*$ given by $T(U_{2,1}) = U_{1,1} U_{2,1}$ and the identity function on other elements of $\mathcal{X}$.
\end{enumerate}
If $U_{1,1} = U_{2,1}$, then we must apply case I when applying a Nielsen transformation. Otherwise, we have a choice of which of the three cases to apply; thus, the process is non-deterministic. If $m = 0$ or $n = 0$, then we cannot apply a Nielsen transformation.

The definition of Nielsen transformations can be motivated as follows. Given $w \in \Gamma^*$, let $|w|$ be the length of $w$. Given a word equation $(U_1, U_2)$ and a solution $S$, we reason by cases. If $|S(U_{1,1})| = |S(U_{2,1})|$, then we must have $S(U_{1,1}) = S(U_{2,1})$, so we can replace the variable $U_{2,1}$ by $U_{1,1}$ throughout the word equation. We can then remove the first variable on both sides of the word equation, since they are now both $U_{1,1}$. This corresponds to case I of the Nielsen transformation. If $|S(U_{1,1})| > |S(U_{2,1})|$, then we must have that $S(U_{2,1})$ is a prefix of $S(U_{1,1})$. Thus, we can replace the variable $U_{1,1}$ by $U_{2,1} U_{1,1}$ throughout the word equation (after this replacement, the new $U_{1,1}$ represents a suffix of the old $U_{1,1}$). We can then remove the first variable on both sides of the word equation, since they are now both $U_{2,1}$. This corresponds to case II of the Nielsen transformation. Similarly, the case where $|S(U_{1,1})| < |S(U_{2,1})|$ corresponds to case III of the Nielsen transformation.

\begin{example}
    The three cases of the Nielsen transformation applied to $(\mathit{XY}, \mathit{ZX})$ yield, respectively,
    \begin{enumerate}[label=\Roman*.]
        \item $(Y,X)$,
        \item $(\mathit{XY},\mathit{ZX})$, and
        \item $(Y,\mathit{ZX})$.
    \end{enumerate}
\end{example}

To describe relative lengths in extended word equations, we will need the notion of a strict weak order.

\begin{definition}
    A \emph{strict partial order} is a binary relation that is irreflexive and transitive. If $<$ is a strict partial order, then we write $a \approx b$ if $a \nless b$ and $b \nless a$. We write $a \lesssim b$ if $a < b$ or $a \approx b$. A \emph{strict weak order} is a strict partial order $<$ such that its corresponding relation $\approx$ is an equivalence relation.
\end{definition}

\begin{definition}
    Let $<$ be a strict weak order. We write $a \lessdot b$ if $a < b$ and there is no $c \notin \{a,b\}$ such that $a < c < b$. If $a < b$ and there is no $c \notin \{a,b\}$ such that $a \lesssim c \lesssim b$, then we say that $a$ is \emph{adjacent} to $b$ (and vice versa).
\end{definition}

If $a$ is adjacent to $b$ with $a < b$, then there is a strict weak order $<'$ obtained by swapping the order of $a$ and $b$ in $<$; that is, $c <' d$ if and only if $(c,d) = (b,a)$ or else $c < d$ and $(c,d) \neq (a,b)$.

Given $m,n \in \mathbb{N}$, let
\begin{align*}
    [m,n] & = \{\ell \in \mathbb{N} : m \le \ell \le n\} \quad\text{and} \\
    [n] & = [1,n].
\end{align*}

Given a set $X$, let $\mathbb{N}[X]$ be the free commutative monoid generated by $X$ (written additively). Given $x \in X$ and $a \in \mathbb{N}[X]$, we write $x \in a$ if the coefficient of $x$ in $a$ is nonzero. Given $a,b \in \mathbb{N}[X]$, we write $a \subseteq b$ if there is some $c \in \mathbb{N}[X]$ such that $b = a+c$; we write $a \subsetneq b$ if $a \subseteq b$ and $a \neq b$. We consider sets of inequalities over $\mathbb{N}[X]$ of the form
\[
    \{a_i \le b_i \mid i \in [m]\} \cup \{c_i < d_i \mid i \in [n]\},
\]
where $a_i, b_i, c_i, d_i \in \mathbb{N}[X]$. We say that such a set of inequalities \emph{has a solution over $\mathbb{Z}_{>0}$} if there is a homomorphism $L : \mathbb{N}[X] \to (\mathbb{N}, +)$ such that $L(x) > 0$ for all $x \in X$ and $L(a_i) \le L(b_i)$ for all $i \in [m]$ and $L(c_i) < L(d_i)$ for all $i \in [n]$.

We use some terminology from graph theory. All graphs in this paper are directed with loops allowed. A \emph{walk} in a graph $G = (V,E)$ is a nonempty sequence of vertices $(v_1, v_2, \dots, v_n)$ such that $(v_i, v_{i+1}) \in E$ for all $i \in [n-1]$. The \emph{length} of a walk $(v_1, v_2, \dots, v_n)$ is $n-1$. A \emph{path} is a walk in which all vertices are distinct. A \emph{cycle} is a walk of nonzero length of the form $(v_1, v_2, \dots, v_{n-1}, v_1)$. A graph is \emph{acyclic} if it contains no cycles; otherwise, it is \emph{cyclic}.

\section{Extended word equations}

It's now time to introduce extended word equations. We start with some motivation before giving a formal definition. Consider the word equation $(\mathit{XXY}, \mathit{ZWZ})$. We want a representation of this word equation that encodes partial information about the relative lengths of the variables. One natural way to do this would be with a diagram like the one in Figure~\ref{fig-xxy-zwz}.

\begin{figure}[H]
        \centering
        \begin{tikzpicture}
            \draw (0,0) rectangle (2,1);
            \draw (2,0) rectangle (5,1);
            \draw (5,0) rectangle (7,1);
            \draw (0,1) rectangle (3,2);
            \draw (3,1) rectangle (6,2);
            \draw (6,1) rectangle (7,2);

            \node at (1.5,1.5) {$X$};
            \node at (4.5,1.5) {$X$};
            \node at (6.5,1.5) {$Y$};
            \node at (1,0.5) {$Z$};
            \node at (3.5,0.5) {$W$};
            \node at (6,0.5) {$Z$};
        \end{tikzpicture}
        \caption{A representation of the word equation $(\mathit{XXY}, \mathit{ZWZ})$}
        \label{fig-xxy-zwz}
\end{figure}
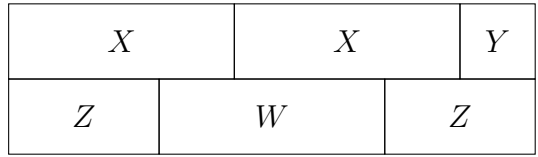

In this diagram, we have assumed that $Z$ is shorter than $X$, which is shorter than $\mathit{ZW}$, which is shorter than $\mathit{XX}$. Of course, it's necessary that $\mathit{XXY}$ be the same length as $\mathit{ZWZ}$, since these are supposed to be equal after applying a solution to both sides. Importantly, we do not assume that a diagram like the one in Figure~\ref{fig-xxy-zwz} is drawn to scale, so it should be considered equivalent to any other diagram making the same assumptions, such as the one in Figure~\ref{fig-equivalent}.

\begin{figure}[H]
        \centering
        \begin{tikzpicture}
            \draw (0,0) rectangle (2.75,1);
            \draw (2.75,0) rectangle (4.25,1);
            \draw (4.25,0) rectangle (7,1);
            \draw (0,1) rectangle (3.1,2);
            \draw (3.1,1) rectangle (6.2,2);
            \draw (6.2,1) rectangle (7,2);

            \node at (1.55,1.5) {$X$};
            \node at (4.65,1.5) {$X$};
            \node at (6.6,1.5) {$Y$};
            \node at (1.375,0.5) {$Z$};
            \node at (3.5,0.5) {$W$};
            \node at (5.625,0.5) {$Z$};
        \end{tikzpicture}
        \caption{An equivalent representation of the word equation $(\mathit{XXY}, \mathit{ZWZ})$}
        \label{fig-equivalent}
\end{figure}
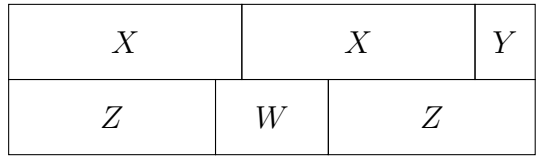

In order to mathematically reason about these diagrammatic representations of word equations, we need to define a data structure that encodes the combinatorial data the diagrams represent. This is what the definition of extended word equations accomplishes.

Given a word equation $(U_1,U_2)$, with
\begin{align*}
    U_1 &= U_{1,1} U_{1,2} \cdots U_{1,m} \quad\text{and} \\
    U_2 &= U_{2,1} U_{2,2} \cdots U_{2,n},
\end{align*}
where $U_{i,j} \in \mathcal{X}$ for all $i,j$, the set of \emph{boundaries} is the set $B = B_1 \cup B_2$, where $B_1 = \{(1,j) \mid 1 \le j \le m\}$ and $B_2 = \{(2,j) \mid 1 \le j \le n\}$.

Intuitively, boundaries represent the right-hand sides of variable instances in diagrams like the ones above. Thus, for the word equation $(\mathit{XXY}, \mathit{ZWZ})$, the boundary $(1,1)$ corresponds to the right border of the first instance of $X$ and the boundary $(2,1)$ corresponds to the right border of the first instance of $Z$. We will impose an order on the boundaries to represent the relative positions of the variable borders. Thus, for the diagrams above, we should have $(1,1) > (2,1)$, since the right-hand side of the first $X$ instance is to the right of the right-hand side of the first $Z$ instance.

\begin{definition}
    Let $X = \{X_1, \dots, X_m\}$ and $Y = \{Y_1, \dots, Y_n\}$ be disjoint totally ordered sets such that $X_i < X_{i+1}$ for $i\in[m-1]$ and $Y_i < Y_{i+1}$ for $i\in[n-1]$. Then, an \emph{interleaving} of $X$ and $Y$ is a strictly weakly ordered set $X \cup Y$ extending the orders on $X$ and $Y$ such that $X_m \approx Y_n$.
\end{definition}

We totally order $B_1$ and $B_2$ so that $(i,j) < (i,j+1)$ for each relevant $i,j$. A \emph{boundary order} $<$ on $B$ is an interleaving of $B_1$ and $B_2$. We will occasionally apply the boundary order to the tuples $(1,0)$ and $(2,0)$ with the convention that $(1,0) \approx (2,0)$ and $(i,0) < (i,1)$ for $i \in \{1,2\}$.

We can now define extended word equations as follows.
\begin{definition}
    An \emph{extended word equation} $(U_1, U_2, <)$ is a word equation $(U_1, U_2)$ together with a boundary order $<$.
\end{definition}

Before concluding this discussion, we need to consider a problematic type of boundary order. Consider the extended word equation represented by the diagram in Figure~\ref{fig-problematic}.

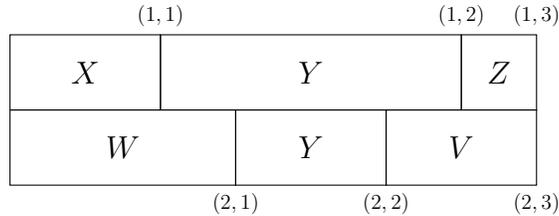
\begin{figure}[H]
        \centering
        \begin{tikzpicture}
            \draw (0,0) rectangle (3,1);
            \draw (3,0) rectangle (5,1);
            \draw (5,0) rectangle (7,1);
            \draw (0,1) rectangle (2,2);
            \draw (2,1) rectangle (6,2);
            \draw (6,1) rectangle (7,2);

            \node at (1,1.5) {$X$};
            \node at (4,1.5) {$Y$};
            \node at (6.5,1.5) {$Z$};
            \node at (1.5,0.5) {$W$};
            \node at (4,0.5) {$Y$};
            \node at (6,0.5) {$V$};

            \node[scale=0.67] at (2,2.25) {$(1,1)$};
            \node[scale=0.67] at (6,2.25) {$(1,2)$};
            \node[scale=0.67] at (7,2.25) {$(1,3)$};
            \node[scale=0.67] at (3,-0.25) {$(2,1)$};
            \node[scale=0.67] at (5,-0.25) {$(2,2)$};
            \node[scale=0.67] at (7,-0.25) {$(2,3)$};
        \end{tikzpicture}
        \caption{A problematic extended word equation}
        \label{fig-problematic}
\end{figure}

As the diagram makes clear, there is something problematic with this boundary order: it implies that the two instances of $Y$ are different lengths. We need some way of ruling out these problematic boundary orders. To that end, we make the following definition.
\begin{definition}
    Let $(U_1,U_2,<)$ be an extended word equation with boundaries $B$. Write
    \begin{align*}
        U_1 &= U_{1,1} U_{1,2} \cdots U_{1,m} \quad\text{and} \\
        U_2 &= U_{2,1} U_{2,2} \cdots U_{2,n},
    \end{align*}
    where $U_{i,j} \in \mathcal{X}$ for all $i,j$. We say that $(U_1,U_2,<)$ is a \emph{coherent extended word equation} if there is a homomorphism $L : \mathcal{X}^* \to (\mathbb{N}, +)$ such that $L(X) > 0$ for all $X \in \mathcal{X}$ and
    \[
        (i,j) < (i',j') \Longleftrightarrow L(U_{i,1} U_{i,2} \cdots U_{i,j}) < L(U_{i',1} U_{i',2} \cdots U_{i',j'}).
    \]
    for all $(i,j), (i',j') \in B$. Otherwise, $(U_1,U_2,<)$ is an \emph{incoherent extended word equation}.
\end{definition}
Intuitively, $L(X)$ is the length of $X$ in a hypothetical solution. One can check whether an extended word equation is coherent using a decision procedure for linear arithmetic (as is done in SMT solvers).

We also define the \emph{dual} of a word equation, which is the word equation obtained by swapping the two sides of the equation.

\begin{definition}
    The \emph{dual} of $(U_1, U_2, <)$ is $(U_2, U_1, <')$, where
    \[
        (i,j) <' (i',j') \Longleftrightarrow (3-i,j) < (3-i',j').
    \]
\end{definition}

\noindent
We say that $(U_1, U_2, <)$ is \emph{nontrivial} if $U_1,U_2 \neq \varepsilon$.

\section{Extended Nielsen transformations} \label{sec-ext-nielsen}

Now, we adapt Nielsen transformations from word equations to extended word equations. We start with an example. Consider the extended word equation from Figure~\ref{fig-xxy-zwz}. We know that $Z$ is shorter than $X$, so we want to split $X$ into $\mathit{ZX}$. This corresponds to a case II Nielsen transformation, which applied to $(\mathit{XXY}, \mathit{ZWZ})$ yields $(\mathit{XZXY},\mathit{WZ})$. The possible boundary orders on $(\mathit{XZXY},\mathit{WZ})$ should be constrained by the boundary order on $(\mathit{XXY}, \mathit{ZWZ})$. Thus, we should expect that an extended Nielsen transformation applied to the extended word equation from Figure~\ref{fig-xxy-zwz} should yield one of the three extended word equations in Figure~\ref{fig-three-word-eqns}.

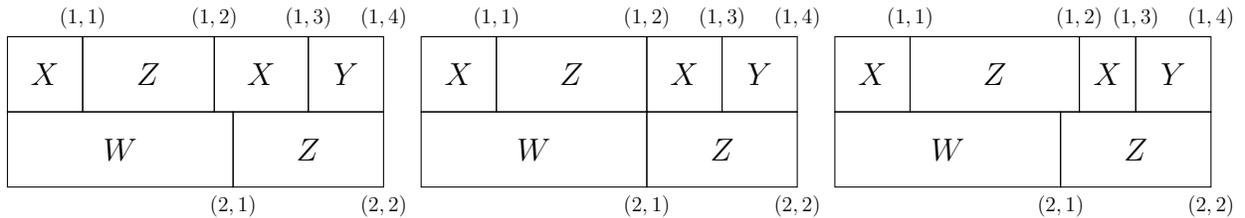
\begin{figure}[H]
        \centering
        \begin{tikzpicture}
            \draw (-5.5,0) rectangle (-2.5,1);
            \draw (-2.5,0) rectangle (-0.5,1);
            \draw (-5.5,1) rectangle (-4.5,2);
            \draw (-4.5,1) rectangle (-2.75,2);
            \draw (-2.75,1) rectangle (-1.5,2);
            \draw (-1.5,1) rectangle (-0.5,2);

            \node at (-5,1.5) {$X$};
            \node at (-3.625,1.5) {$Z$};
            \node at (-2.125,1.5) {$X$};
            \node at (-1,1.5) {$Y$};
            \node at (-4,0.5) {$W$};
            \node at (-1.5,0.5) {$Z$};

            \node[scale=0.67] at (-4.5,2.25) {$(1,1)$};
            \node[scale=0.67] at (-2.75,2.25) {$(1,2)$};
            \node[scale=0.67] at (-1.5,2.25) {$(1,3)$};
            \node[scale=0.67] at (-0.5,2.25) {$(1,4)$};
            \node[scale=0.67] at (-2.5,-0.25) {$(2,1)$};
            \node[scale=0.67] at (-0.5,-0.25) {$(2,2)$};
        
            \draw (0,0) rectangle (3,1);
            \draw (3,0) rectangle (5,1);
            \draw (0,1) rectangle (1,2);
            \draw (1,1) rectangle (3,2);
            \draw (3,1) rectangle (4,2);
            \draw (4,1) rectangle (5,2);

            \node at (0.5,1.5) {$X$};
            \node at (2,1.5) {$Z$};
            \node at (3.5,1.5) {$X$};
            \node at (4.5,1.5) {$Y$};
            \node at (1.5,0.5) {$W$};
            \node at (4,0.5) {$Z$};

            \node[scale=0.67] at (1,2.25) {$(1,1)$};
            \node[scale=0.67] at (3,2.25) {$(1,2)$};
            \node[scale=0.67] at (4,2.25) {$(1,3)$};
            \node[scale=0.67] at (5,2.25) {$(1,4)$};
            \node[scale=0.67] at (3,-0.25) {$(2,1)$};
            \node[scale=0.67] at (5,-0.25) {$(2,2)$};

            \draw (5.5,0) rectangle (8.5,1);
            \draw (8.5,0) rectangle (10.5,1);
            \draw (5.5,1) rectangle (6.5,2);
            \draw (6.5,1) rectangle (8.75,2);
            \draw (8.75,1) rectangle (9.5,2);
            \draw (9.5,1) rectangle (10.5,2);

            \node at (6,1.5) {$X$};
            \node at (7.625,1.5) {$Z$};
            \node at (9.125,1.5) {$X$};
            \node at (10,1.5) {$Y$};
            \node at (7,0.5) {$W$};
            \node at (9.5,0.5) {$Z$};

            \node[scale=0.67] at (6.5,2.25) {$(1,1)$};
            \node[scale=0.67] at (8.75,2.25) {$(1,2)$};
            \node[scale=0.67] at (9.5,2.25) {$(1,3)$};
            \node[scale=0.67] at (10.5,2.25) {$(1,4)$};
            \node[scale=0.67] at (8.5,-0.25) {$(2,1)$};
            \node[scale=0.67] at (10.5,-0.25) {$(2,2)$};
        \end{tikzpicture}
        \caption{Three extended word equations}
        \label{fig-three-word-eqns}
\end{figure}

Making this notion of an extended Nielsen transformation precise requires some work. Let $(U_1, U_2, <)$ be a coherent extended word equation. Our goal is to define what it means for some extended word equation $(U'_1, U'_2, <')$ to be an extended Nielsen transformation of $(U_1, U_2, <)$. Specifying $(U'_1, U'_2)$ is just a matter of applying the right kind of Nielsen transformation:
\begin{enumerate}[label=\Roman*.]
    \item If $(1,1) \approx (2,1)$, then $(U'_1, U'_2)$ needs to be the case I Nielsen transformation of $(U_1, U_2)$.
    \item If $(1,1) > (2,1)$, then $(U'_1, U'_2)$ needs to be the case II Nielsen transformation of $(U_1, U_2)$.
    \item If $(1,1) < (2,1)$, then $(U'_1, U'_2)$ needs to be the case III Nielsen transformation of $(U_1, U_2)$.
\end{enumerate}
It remains to specify the boundary order $<'$.

Suppose first that we are in case I. In this case, an extended Nielsen transformation should behave as in the example in Figure~\ref{fig-extended-case-1}.

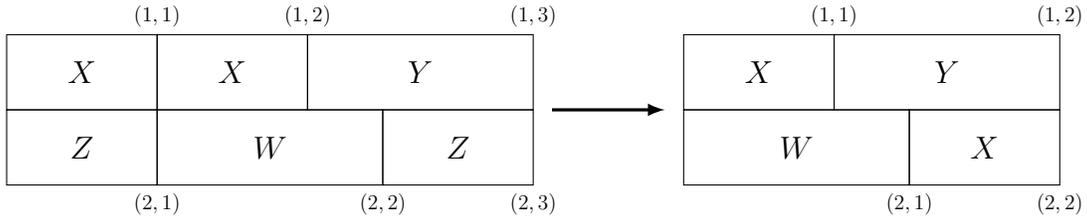
\begin{figure}[H]
        \centering
        \begin{tikzpicture}
            \draw (0,0) rectangle (2,1);
            \draw (2,0) rectangle (5,1);
            \draw (5,0) rectangle (7,1);
            \draw (0,1) rectangle (2,2);
            \draw (2,1) rectangle (4,2);
            \draw (4,1) rectangle (7,2);

            \node at (1,1.5) {$X$};
            \node at (3,1.5) {$X$};
            \node at (5.5,1.5) {$Y$};
            \node at (1,0.5) {$Z$};
            \node at (3.5,0.5) {$W$};
            \node at (6,0.5) {$Z$};

            \node[scale=0.67] at (2,2.25) {$(1,1)$};
            \node[scale=0.67] at (4,2.25) {$(1,2)$};
            \node[scale=0.67] at (7,2.25) {$(1,3)$};
            \node[scale=0.67] at (2,-0.25) {$(2,1)$};
            \node[scale=0.67] at (5,-0.25) {$(2,2)$};
            \node[scale=0.67] at (7,-0.25) {$(2,3)$};

            \draw (9,0) rectangle (12,1);
            \draw (12,0) rectangle (14,1);
            \draw (9,1) rectangle (11,2);
            \draw (11,1) rectangle (14,2);

            \node at (10,1.5) {$X$};
            \node at (12.5,1.5) {$Y$};
            \node at (10.5,0.5) {$W$};
            \node at (13,0.5) {$X$};

            \node[scale=0.67] at (11,2.25) {$(1,1)$};
            \node[scale=0.67] at (14,2.25) {$(1,2)$};
            \node[scale=0.67] at (12,-0.25) {$(2,1)$};
            \node[scale=0.67] at (14,-0.25) {$(2,2)$};

            \draw [arrow] (7.25,1) -- (8.75,1);
        \end{tikzpicture}
        \caption{How an extended Nielsen transformation should behave in case I}
        \label{fig-extended-case-1}
\end{figure}
We see that, in this case, the boundary $(i,j)$ in $(U'_1, U'_2)$ corresponds to the boundary $(i,j+1)$ in $(U_1, U_2)$. Thus, in case I, the boundary order $<'$ is given by
\[
    (i,j) <' (i',j') \Longleftrightarrow (i, j+1) < (i', j'+1).
\]

Next, suppose that we are in case II. In this case, an extended Nielsen transformation should behave as in the example in Figure~\ref{fig-extended-case-2}.

\begin{figure}[H]
        \centering
        \begin{tikzpicture}
            \draw (0,0) rectangle (2,1);
            \draw (2,0) rectangle (5,1);
            \draw (5,0) rectangle (7,1);
            \draw (0,1) rectangle (3,2);
            \draw (3,1) rectangle (6,2);
            \draw (6,1) rectangle (7,2);

            \node at (1.5,1.5) {$X$};
            \node at (4.5,1.5) {$X$};
            \node at (6.5,1.5) {$Y$};
            \node at (1,0.5) {$Z$};
            \node at (3.5,0.5) {$W$};
            \node at (6,0.5) {$Z$};

            \node[scale=0.67] at (3,2.25) {$(1,1)$};
            \node[scale=0.67] at (6,2.25) {$(1,2)$};
            \node[scale=0.67] at (7,2.25) {$(1,3)$};
            \node[scale=0.67] at (2,-0.25) {$(2,1)$};
            \node[scale=0.67] at (5,-0.25) {$(2,2)$};
            \node[scale=0.67] at (7,-0.25) {$(2,3)$};

            \draw (9,0) rectangle (12,1);
            \draw (12,0) rectangle (14,1);
            \draw (9,1) rectangle (10,2);
            \draw (10,1) rectangle (12.25,2);
            \draw (12.25,1) rectangle (13,2);
            \draw (13,1) rectangle (14,2);

            \node at (9.5,1.5) {$X$};
            \node at (11.125,1.5) {$Z$};
            \node at (12.625,1.5) {$X$};
            \node at (13.5,1.5) {$Y$};
            \node at (10.5,0.5) {$W$};
            \node at (13,0.5) {$Z$};

            \node[scale=0.67] at (10,2.25) {$(1,1)$};
            \node[scale=0.67] at (12.25,2.25) {$(1,2)$};
            \node[scale=0.67] at (13,2.25) {$(1,3)$};
            \node[scale=0.67] at (14,2.25) {$(1,4)$};
            \node[scale=0.67] at (12,-0.25) {$(2,1)$};
            \node[scale=0.67] at (14,-0.25) {$(2,2)$};

            \draw [arrow] (7.25,1) -- (8.75,1);
        \end{tikzpicture}
        \caption{How an extended Nielsen transformation should behave in case II}
        \label{fig-extended-case-2}
\end{figure}
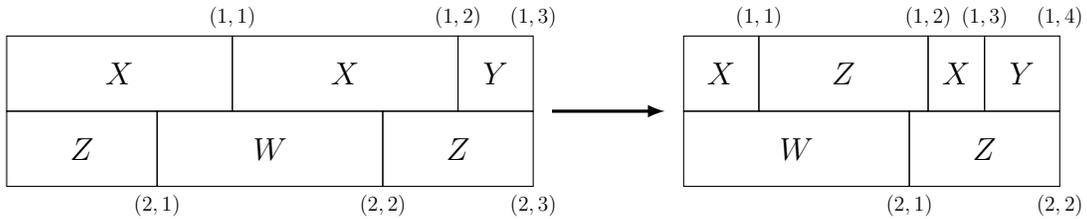

Here, it takes more work to describe how boundaries of $(U'_1, U'_2)$ correspond to boundaries of $(U_1, U_2)$. This is because new variable instances have been introduced in $(U'_1, U'_2)$, and this throws off the indexing of the boundaries. To fix this, we need a map $\mu$ which ensures that the boundary $(i,j)$ in $(U_1, U_2)$ corresponds to the boundary $(i, \mu(i,j))$ in $(U'_1, U'_2)$. The right definition is
\[
    \mu(i,j) = j+|\{(i,j') \mid j' \le j \ \text{and} \ U_{i,j'} = U_{1,1}\}|-1.
\]
We also define a map $\nu$ so that $(i,j) \mapsto (i, \nu(i,j))$ is a left inverse of $(i,j) \mapsto (i, \mu(i,j))$:\footnote{The inverse property is proven later in Lemma~\ref{lem-left-inverse}.}
\[
    \nu(i,j) = j-\left|\left\{(i,j') \mathrel{}\middle|\mathrel{} j' \le j \ \text{and} \ U'_{i,j'} = U_{1,1}\right\}\right| + 1.
\]

Let $B = B_1 \cup B_2$ be the set of boundaries of $(U_1, U_2)$ and let $B' = B'_1 \cup B'_2$ be the set of boundaries of $(U'_1, U'_2)$. Let $B^- = B \setminus \{(2,1)\}$, which is the set of boundaries of $(U_1, U_2)$ that correspond to a boundary of $(U'_1, U'_2)$. Let $B'^- = \{(i,\mu(i,j)) \mid (i,j) \in B^-\}$, which is the set of boundaries of $(U'_1, U'_2)$ that came from a boundary of $(U_1, U_2)$. Finally, let $B'^+ = B' \setminus B'^-$, which is the set of new boundaries of $(U'_1, U'_2)$.

\begin{example}
    For the pair of extended word equations in Figure~\ref{fig-extended-case-2}, we have
    \begin{align*}
        B^- &= \{(1,1), (1,2), (1,3), (2,2), (2,3)\} \\
        B'^- &= \{(1,1), (1,3), (1,4), (2,1), (2,2)\} \\
        B'^+ &= \{(1,2)\}.
    \end{align*}
\end{example}
\noindent
It is not hard to see that $\mu$ is at least 1 on $B^-$ and that $\nu$ is at least 1 on $B'$.

Now, with these definitions in hand, we can say that the boundary order $<'$ is any interleaving of $B'_1$ and $B'_2$ such that for all $(i,j), (i',j') \in B'^-$, we have
\[
    (i,j) <' (i',j') \Longleftrightarrow (i,\nu(i,j)) < (i',\nu(i',j')).
\]
The intuition here is that the order on the boundaries in $B'^-$ is fixed by the order on the corresponding boundaries of $B^-$. For the remaining boundaries (i.e., those in $B'^+$), we have some choice in how to order them.

Case III is handled similarly to case II. Formally, we reduce to case II using duality. This concludes the motivation for the definition of extended Nielsen transformations.
\begin{definition}
    An \emph{extended Nielsen transformation} of a nontrivial coherent extended word equation $(U_1, U_2, <)$ is an extended word equation $(U'_1, U'_2, <')$ given by one of the following three cases:
    \begin{enumerate}[label=\Roman*.]
        \item If $(1,1) \approx (2,1)$, then $(U'_1, U'_2)$ is the case I Nielsen transformation of $(U_1, U_2)$, and $<'$ is given by
        \[
            (i,j) <' (i',j') \Longleftrightarrow (i, j+1) < (i', j'+1).
        \]
        \item If $(1,1) > (2,1)$, then $(U'_1, U'_2)$ is the case II Nielsen transformation of $(U_1, U_2)$. Let $B'_1$, $B'_2$, $B'^-$, and $\nu$ be defined as in the above discussion. Then, $<'$ is an interleaving of $B'_1$ and $B'_2$ such that for all $(i,j), (i',j') \in B'^-$, we have
        \[
            (i,j) <' (i',j') \Longleftrightarrow (i,\nu(i,j)) < (i',\nu(i',j')).
        \]
        \item If $(1,1) < (2,1)$, then $(U'_1, U'_2, <')$ is obtained by applying a case II extended Nielsen transformation to the dual of $(U_1, U_2, <)$ and then taking the dual again.
    \end{enumerate}
    If $(U'_1, U'_2, <')$ is coherent, this transformation is a \emph{coherent extended Nielsen transformation}; otherwise, it is an \emph{incoherent extended Nielsen transformation}.
\end{definition}

\subsection{Termination of extended Nielsen transformations}

We can finally ask the central question of this paper: if we repeatedly apply coherent extended Nielsen transformations to $(U_1, U_2, <)$, will this process always terminate? There is an ambiguity in this question depending on how the procedure is implemented. The simplest way to implement it is to simply apply a coherent extended Nielsen transformation repeatedly, only keeping track of the current extended word equation. If this process always terminates, we say that $(U_1,U_2, <)$ is \emph{strongly terminating}. A more sophisticated procedure remembers which extended word equations have appeared so far, and only applies a coherent extended Nielsen transformation if it yields a \emph{new} extended word equation. If this process always terminates, we say that $(U_1,U_2, <)$ is \emph{weakly terminating}.

We now describe this more formally. Let $\mathcal{G}$ be the directed graph $(\mathcal{V}, \mathcal{E})$, where $\mathcal{V}$ is the set of all coherent extended word equations (over some fixed set of variables $\mathcal{X}$), and $\mathcal{E}$ is the relation that holds between $(U_1,U_2,<)$ and $(U'_1,U'_2,<')$ when the latter is the result of applying an extended Nielsen transformation to the former. Then, let $\mathcal{G}_{U_1,U_2,<}$ be the subgraph of $\mathcal{G}$ that can be reached from $(U_1,U_2,<)$.\footnote{An analogous graph has been studied in the setting of word equations without length constraints \cite{structure}.}

A coherent extended word equation $(U_1, U_2, <)$ is \emph{strongly terminating} if there are no infinite walks in $\mathcal{G}_{U_1,U_2,<}$. A coherent extended word equation $(U_1, U_2, <)$ is \emph{weakly terminating} if there are no infinite paths in $\mathcal{G}_{U_1,U_2,<}$. By K\H{o}nig's lemma, $(U_1,U_2,<)$ is strongly terminating if and only if $\mathcal{G}_{U_1,U_2,<}$ is finite and acyclic; $(U_1,U_2,<)$ is weakly terminating if and only if $\mathcal{G}_{U_1,U_2,<}$ is finite.

For the remainder of this paper, we will only be concerned with strong termination, both because it is simpler to study and because it more closely resembles the implementations in the SMT solvers that motivated this investigation. Thus, we hereafter write \emph{terminating} instead of \emph{strongly terminating}.

\section{The cut graph} \label{sec-cut-graph}

Given a single extended word equation $(U_1, U_2, <)$, we define a corresponding graph whose properties turn out to be the key to understanding whether $(U_1, U_2, <)$ is terminating.
In an extended word equation, some variable instances overlap each other. When this happens, we say that their boundaries \emph{cut} each other.
\begin{definition}
    Given an extended word equation $(U_1, U_2, <)$, we say that the boundary $(i,j)$ \emph{cuts} $(i',j')$ if
    \begin{align*}
        (i,j) &\neq (i',j'), \\
        (i,j-1) &< (i',j'), \quad\text{and} \\
        (i',j'-1) &< (i,j).
    \end{align*}
\end{definition}

\noindent
We also define a relation that holds between two instances of the same variable in different positions. When this happens, we say that their boundaries \emph{mirror} each other.
\begin{definition}
    Given an extended word equation $(U_1, U_2, <)$, we say that the boundary $(i,j)$ \emph{mirrors} $(i',j')$ if $U_{i,j} = U_{i',j'}$ and either $(i,j) \not\approx (i',j')$ or $(i,j-1) \not\approx (i',j'-1)$
\end{definition}

\noindent
Note that if $(U_1, U_2, <)$ is coherent, then $(i,j)$ mirrors $(i',j')$ if and only if $U_{i,j} = U_{i',j'}$ and $(i,j) \not\approx (i',j')$. We can now define the graph mentioned earlier.

\begin{definition}
    Given an extended word equation $(U_1, U_2, <)$ with boundaries $B$, we define its \emph{cut graph} $G_{U_1,U_2,<} = (B,E)$ as follows:
    \[
        E = \{((i,j), (i',j')) \mid (i,j) \ \text{cuts some} \ (i'',j'') \ \text{and} \ (i'',j'') \ \text{mirrors} \ (i',j')\}.
    \]
\end{definition}

\noindent
Figure~\ref{fig-cut-graph} shows the cut graph of the extended word equation from Figure~\ref{fig-xxy-zwz}.

\begin{figure}[H]
        \centering
        \begin{tikzpicture}
            \node[vertex] (x1) at (0,3) {$(1,1)$};
            \node[vertex] (x2) at (3,3) {$(1,2)$};
            \node[vertex] (y) at (6,3) {$(1,3)$};
            \node[vertex] (z1) at (0,0) {$(2,1)$};
            \node[vertex] (w) at (3,0) {$(2,2)$};
            \node[vertex] (z2) at (6,0) {$(2,3)$};

            \draw[edge] (x1) to[bend left = 10] (z2);
            \draw[edge] (x2) to[bend left = 10] (z1);
            \draw[edge] (y) to (z1);
            \draw[edge] (z1) to[bend left = 10] (x2);
            \draw[edge] (w) to (x1);
            \draw[edge] (w) to (x2);
            \draw[edge] (z2) to[bend left = 10] (x1);
        \end{tikzpicture}
        \caption{The cut graph of the extended word equation from Figure~\ref{fig-xxy-zwz}}
        \label{fig-cut-graph}
\end{figure}
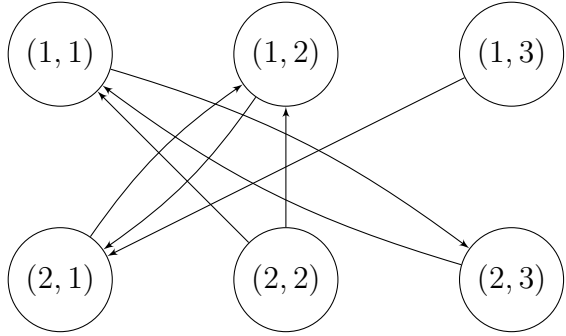

\section{A sufficient condition for termination}

The purpose of this section is to prove the following theorem.

\begin{theorem} \label{thm-sufficient}
    Let $(U_1,U_2,<)$ be a coherent extended word equation. If $G_{U_1,U_2,<}$ is acyclic, then $(U_1,U_2,<)$ is terminating. Furthermore, $(U_1,U_2,<)$ terminates after the application of at most $2^N$ coherent extended Nielsen transformations, where $N = |U_1| + |U_2|$.
\end{theorem}

\noindent
We need the following definitions.
\begin{definition}
    Given an extended word equation $(U_1, U_2, <)$ such that $G_{U_1, U_2, <} = (B,E)$ is acyclic, the \emph{fecundity} of the boundary $(i,j)$ is recursively defined as
    \[
        f_{U_1,U_2,<}(i,j) = 1 + \max_{\substack{(i'',j'') \in B \\ (i,j) \ \text{cuts} \ (i'',j'')}} \sum_{\substack{(i',j') \in B \\ (i'',j'') \ \text{mirrors} \ (i',j')}} f_{U_1,U_2,<}(i',j').
    \]
\end{definition}
\noindent
Note that the definition of $f_{U_1,U_2,<}(i,j)$ only depends on the values of $f_{U_1,U_2,<}(i',j')$ where $((i,j),(i',j')) \in E$, so the recursion is well-founded if $G_{U_1, U_2, <}$ is acyclic.
\begin{definition}
    Given an extended word equation $(U_1, U_2, <)$ such that $G_{U_1, U_2, <} = (B,E)$ is acyclic, its \emph{measure} is
    \[
        m(U_1, U_2, <) = \sum_{(i,j) \in B} f_{U_1,U_2,<}(i,j).
    \]
\end{definition}
We will prove Theorem~\ref{thm-sufficient} by showing that with each application of a coherent extended Nielsen transformation, the measure decreases. To that end, let $(U'_1,U'_2,<')$ be the result of applying a coherent extended Nielsen transformation to $(U_1,U_2,<)$. Write $G_{U_1, U_2, <} = (B,E)$ and $G_{U'_1, U'_2, <'} = (B',E')$. Assume that $G_{U_1, U_2, <}$ is acyclic. We aim to show that $m(U'_1, U'_2, <') < m(U_1, U_2, <)$.

Before going into the details, a brief description of the proof strategy is in order. The difficult part is proving the measure decreases when a case II extended Nielsen transformation is applied. In this case, let $B^-$, $B'^-$, and $B'^+$ be defined as in Section~\ref{sec-ext-nielsen}. Lemma~\ref{lem-walk-case2} shows that $G_{U'_1, U'_2, <'}$ is acyclic, so $f_{U'_1,U'_2,<'}$ is well-defined. Lemma~\ref{lem-mu-decrease} implies that the fecundity of every boundary in $B'^-$ is at most the fecundity of its corresponding boundary in $B^-$. Lemma~\ref{lem-less-21} shows that the sum of the fecundities of the boundaries in $B'^+$ is strictly less than the fecundity of $(2,1)$ in $(U_1,U_2,<)$. The boundary $(2,1)$ in $(U_1,U_2,<)$ does not correspond to any boundary in $(U'_1,U'_2,<')$. Thus, in the process of going from $(U_1,U_2,<)$ to $(U'_1,U'_2,<')$, we lose the boundary $(2,1)$, add some boundaries whose fecundities sum to strictly less than that of $(2,1)$ in $(U_1,U_2,<)$, and retain some boundaries whose fecundities do not increase. Therefore, $m(U'_1, U'_2, <') < m(U_1, U_2, <)$.

First, we prove a few lemmas for when a case I extended Nielsen transformation is applied.

\begin{lemma} \label{lem-mirror-case1}
    Suppose $(U'_1,U'_2,<')$ is the result of applying a case I coherent extended Nielsen transformation to $(U_1,U_2,<)$. If $(i,j)$ mirrors $(i',j')$ with respect to $(U'_1,U'_2,<')$, then either $(i,j+1)$ mirrors $(i',j'+1)$ with respect to $(U_1,U_2,<)$ or $\{U_{i,j+1}, U_{i',j'+1}\} = \{U_{1,1}, U_{2,1}\}$.
\end{lemma}
\begin{proof}
    Suppose $(i,j)$ mirrors $(i',j')$ with respect to $(U'_1,U'_2,<')$. Then, $(i,j) \not\approx' (i',j')$, so $(i,j+1) \not\approx (i',j'+1)$. If $U_{i,j+1} = U_{i',j'+1}$, then $(i,j+1)$ mirrors $(i',j'+1)$ with respect to $(U_1,U_2,<)$, and we are done. So suppose $U_{i,j+1} \neq U_{i',j'+1}$.

    Now, we have $U'_{i,j} = T(U_{i,j+1})$ and $U'_{i',j'} = T(U_{i',j'+1})$, where $T$ is the endomorphism on $\mathcal{X}^*$ given by $T(U_{2,1}) = U_{1,1}$ and the identity function on other elements of $\mathcal{X}$. Since $(i,j)$ mirrors $(i',j')$ with respect to $(U'_1,U'_2,<')$, we have
    \[
        T(U_{i,j+1}) = U'_{i,j} = U'_{i',j'} = T(U_{i',j'+1}).
    \]
    Thus, $\{U_{i,j+1}, U_{i',j'+1}\} = \{U_{1,1}, U_{2,1}\}$, as desired.
\end{proof}

\begin{lemma} \label{lem-walk-case1}
    Suppose $(U'_1,U'_2,<')$ is the result of applying a case I coherent extended Nielsen transformation to $(U_1,U_2,<)$. If $G_{U_1,U_2,<}$ is acyclic, then so is $G_{U'_1,U'_2,<'}$.
\end{lemma}
\begin{proof}
    It suffices to show that if $((i,j),(i',j')) \in E'$, then there is a walk of length at least 1 from $(i,j+1)$ to $(i',j'+1)$ in $G_{U_1, U_2, <}$. So suppose $((i,j),(i',j')) \in E'$. Then, with respect to $(U'_1,U'_2,<')$, $(i,j)$ cuts some $(i'',j'') \in B'$ and $(i'',j'')$ mirrors $(i',j')$. It follows that with respect to $(U_1,U_2,<)$, $(i,j+1)$ cuts $(i'',j''+1)$. If $(i'',j''+1)$ mirrors $(i',j'+1)$ with respect to $(U_1,U_2,<)$, then $((i,j+1),(i',j'+1)) \in E$, and we are done. Otherwise, $\{U_{i'',j''+1}, U_{i',j'+1}\} = \{U_{1,1}, U_{2,1}\}$ by Lemma~\ref{lem-mirror-case1}.

    If $U_{i'',j''+1} = U_{1,1}$ and $U_{i',j'+1} = U_{2,1}$, then $(i'',j''+1)$ mirrors $(1,1)$ with respect to $(U_1,U_2,<)$. Thus, $((i,j+1),(1,1)) \in E$. Also, with respect to $(U_1,U_2,<)$, $(1,1)$ cuts $(2,1)$, which mirrors $(i',j'+1)$. Thus, $((1,1),(i',j'+1)) \in E$. The walk $((i,j+1),(1,1),(i',j'+1))$ is our desired walk in $G_{U_1, U_2, <}$.

    Similarly, if $U_{i'',j''+1} = U_{2,1}$ and $U_{i',j'+1} = U_{1,1}$, then $((i,j+1),(2,1),(i',j'+1))$ is our desired walk in $G_{U_1, U_2, <}$.
\end{proof}

\begin{lemma} \label{lem-case1}
    Suppose $(U'_1,U'_2,<')$ is the result of applying a case I coherent extended Nielsen transformation to $(U_1,U_2,<)$. Then, for every $(i,j) \in B'$, we have $f_{U'_1,U'_2,<'}(i,j) \le \break f_{U_1,U_2,<}(i,j+1)$.
\end{lemma}
\begin{proof}
    First, $G_{U'_1,U'_2,<'}$ is acyclic by Lemma~\ref{lem-walk-case1}, so $f_{U'_1,U'_2,<'}$ is well-defined. We may inductively assume that the lemma is true for every $(i',j') \in B'$ such that $((i,j),(i',j')) \in E'$.

    We have
    \[
        f_{U'_1,U'_2,<'}(i,j) = 1 + \sum_{\substack{(i',j') \in B' \\ (i'',j'') \ \text{mirrors} \ (i',j') \\ \text{w.r.t.} \ (U'_1,U'_2,<')}} f_{U'_1,U'_2,<'}(i',j')
    \]
    for some $(i'',j'') \in B'$ such that $(i,j)$ cuts $(i'',j'')$ with respect to $(U'_1,U'_2,<')$. Then, $(i,j+1)$ cuts $(i'',j''+1)$ with respect to $(U_1,U_2,<)$.

    First, suppose $U_{i'',j''+1} \notin \{U_{1,1}, U_{2,1}\}$. Then,
    \begin{align*}
        f_{U'_1,U'_2,<'}(i,j) &= 1 + \sum_{\substack{(i',j') \in B' \\ (i'',j'') \ \text{mirrors} \ (i',j') \\ \text{w.r.t.} \ (U'_1,U'_2,<')}} f_{U'_1,U'_2,<'}(i',j') \\
        &\le 1 + \sum_{\substack{(i',j') \in B' \\ (i'',j'') \ \text{mirrors} \ (i',j') \\ \text{w.r.t.} \ (U'_1,U'_2,<')}} f_{U_1,U_2,<}(i',j'+1) \\
        &\le 1 + \sum_{\substack{(i',j') \in B' \\ (i'',j''+1) \ \text{mirrors} \ (i',j'+1) \\ \text{w.r.t.} \ (U_1,U_2,<)}} f_{U_1,U_2,<}(i',j'+1) \quad \text{by Lemma~\ref{lem-mirror-case1}} \\
        &= 1 + \sum_{\substack{(i',j') \in B \\ (i'',j''+1) \ \text{mirrors} \ (i',j') \\ \text{w.r.t.} \ (U_1,U_2,<)}} f_{U_1,U_2,<}(i',j') \le f_{U_1,U_2,<}(i,j+1).
    \end{align*}

    Next, suppose $U_{i'',j''+1} = U_{1,1}$. Then,
    \begin{align*}
        f_{U'_1,U'_2,<'}(i,j) = 1 &+ \sum_{\substack{(i',j') \in B' \\ (i'',j'') \ \text{mirrors} \ (i',j') \\ \text{w.r.t.} \ (U'_1,U'_2,<')}} f_{U'_1,U'_2,<'}(i',j') \\
        \le 1 &+ \sum_{\substack{(i',j') \in B' \\ (i'',j'') \ \text{mirrors} \ (i',j') \\ \text{w.r.t.} \ (U'_1,U'_2,<')}} f_{U_1,U_2,<}(i',j'+1) \\
        \le 1 &+ \sum_{\substack{(i',j') \in B' \\ (i'',j''+1) \ \text{mirrors} \ (i',j'+1) \\ \text{w.r.t.} \ (U_1,U_2,<)}} f_{U_1,U_2,<}(i',j'+1) \\
        &+ \sum_{\substack{(i',j') \in B' \\ U_{i',j'+1} = U_{2,1}}} f_{U_1,U_2,<}(i',j'+1) \quad \text{by Lemma~\ref{lem-mirror-case1}} \\
        = 1 &+ \sum_{\substack{(i',j') \in B \setminus \{(1,1),(2,1)\} \\ (i'',j''+1) \ \text{mirrors} \ (i',j') \\ \text{w.r.t.} \ (U_1,U_2,<)}} f_{U_1,U_2,<}(i',j') + \sum_{\substack{(i',j') \in B \\ (2,1) \ \text{mirrors} \ (i',j') \\ \text{w.r.t} \ (U_1,U_2,<)}} f_{U_1,U_2,<}(i',j') \\
        < 1 &+ \sum_{\substack{(i',j') \in B \setminus \{(1,1)\} \\ (i'',j''+1) \ \text{mirrors} \ (i',j') \\ \text{w.r.t.} \ (U_1,U_2,<)}} f_{U_1,U_2,<}(i',j') + f_{U_1,U_2,<}(1,1) \le f_{U_1,U_2,<}(i,j+1).
    \end{align*}

    Finally, if $U_{i'',j''+1} = U_{2,1}$, then we similarly have
    \[
        f_{U'_1,U'_2,<'}(i,j) < 1 + \sum_{\substack{(i',j') \in B \setminus \{(2,1)\} \\ (i'',j''+1) \ \text{mirrors} \ (i',j') \\ \text{w.r.t.} \ (U_1,U_2,<)}} f_{U_1,U_2,<}(i',j') + f_{U_1,U_2,<}(2,1) \le f_{U_1,U_2,<}(i,j+1). \qedhere
    \]
\end{proof}

If a case I extended Nielsen transformation is not applied, we may assume without loss of generality that a case II extended Nielsen transformation is applied, since cases II and III are symmetric. We need several lemmas for this case. Let $\mu$, $\nu$, $B^-$, $B'^-$, and $B'^+$ be defined as in Section~\ref{sec-ext-nielsen}.

\begin{lemma} \label{lem-mu-preserve}
    Suppose $(U'_1,U'_2,<')$ is the result of applying a case II coherent extended Nielsen transformation to $(U_1,U_2,<)$. Then, for every $(i,j) \in B^-$, we have $U_{i,j} = U'_{i,\mu(i,j)}$.
\end{lemma}
\begin{proof}
    We have $U'_i = T(U_i)^+$, where $T$ is the endomorphism on $\mathcal{X}^*$ given by $T(U_{1,1}) = U_{2,1} U_{1,1}$ and the identity function on other elements of $\mathcal{X}$. We have
    \[
        T(U_i) = T(U_{i,1} \cdots U_{i,j-1}) T(U_{i,j}) T(U_{i,j+1} \cdots).
    \]
    Also,
    \[
        |T(U_{i,1} \cdots U_{i,j-1})| = j-1 + |\{(i,j') \mid j' \le j-1 \ \text{and} \ U_{i,j'} = U_{1,1}\}|,
    \]
    where $\left|\cdot\right|$ is the length function on $\mathcal{X}^*$.
    
    Suppose first that $U_{i,j} \neq U_{1,1}$. Then,
    \[
        T(U_i) = T(U_{i,1} \cdots U_{i,j-1}) U_{i,j} T(U_{i,j+1} \cdots).
    \]
    Thus, $U_{i,j}$ is equal to the variable in position
    \[
        j + |\{(i,j') \mid j' \le j-1 \ \text{and} \ U_{i,j'} = U_{1,1}\}| = j + |\{(i,j') \mid j' \le j \ \text{and} \ U_{i,j'} = U_{1,1}\}| = \mu(i,j) + 1
    \]
    of $T(U_i)$. Hence, $U'_{i,\mu(i,j)} = U_{i,j}$.

    On the other hand, suppose that $U_{i,j} = U_{1,1}$. Then,
    \[
        T(U_i) = T(U_{i,1} \cdots U_{i,j-1}) U_{2,1} U_{1,1} T(U_{i,j+1} \cdots).
    \]
    Thus, $U_{i,j}$ is equal to the variable in position
    \[
        j + |\{(i,j') \mid j' \le j-1 \ \text{and} \ U_{i,j'} = U_{1,1}\}| + 1 = j + |\{(i,j') \mid j' \le j \ \text{and} \ U_{i,j'} = U_{1,1}\}| = \mu(i,j) + 1
    \]
    of $T(U_i)$. Hence, $U'_{i,\mu(i,j)} = U_{i,j}$.
\end{proof}

\begin{lemma} \label{lem-nu-plus-2}
    Suppose $(U'_1,U'_2,<')$ is the result of applying a case II coherent extended Nielsen transformation to $(U_1,U_2,<)$. If $(i,j) \in B'^+$, then $U'_{i,j} = U_{2,1}$ and $U'_{i,j+1} = U_{1,1}$.
\end{lemma}
\begin{proof}
    Since $(i,j) \in B'^+$, we have $j+1 = \mu(i,k)$ for some $k$ such that $U_{i,k} = U_{1,1}$. By Lemma~\ref{lem-mu-preserve},
    \[
        U'_{i,j+1} = U'_{i,\mu(i,k)} = U_{i,k} = U_{1,1}.
    \]
    Recall that $U'_i = T(U_i)^+$, where $T$ is the endomorphism on $\mathcal{X}^*$ given by $T(U_{1,1}) = U_{2,1} U_{1,1}$ and the identity function on other elements of $\mathcal{X}$. Since $U'_{i,j}$ immediately precedes $U'_{i,j+1} = U_{1,1}$, we have $U'_{i,j} = U_{2,1}$.
\end{proof}

\begin{lemma} \label{lem-nu-j-plus-1}
    Suppose $(U'_1,U'_2,<')$ is the result of applying a case II coherent extended Nielsen transformation to $(U_1,U_2,<)$. If $(i,j) \in B'^+$, then $\nu(i,j) = \nu(i,j+1)$.
\end{lemma}
\begin{proof}
    Lemma~\ref{lem-nu-plus-2} implies
    \[
        \left|\left\{(i,j') \mathrel{}\middle|\mathrel{} j' \le j+1 \ \text{and} \ U'_{i,j'} = U_{1,1}\right\}\right| - \left|\left\{(i,j') \mathrel{}\middle|\mathrel{} j' \le j \ \text{and} \ U'_{i,j'} = U_{1,1}\right\}\right| = 1.
    \]
    Hence,
    \begin{align*}
        \nu(i,j) &= j - \left|\left\{(i,j') \mathrel{}\middle|\mathrel{} j' \le j \ \text{and} \ U'_{i,j'} = U_{1,1}\right\}\right| + 1 \\
        &= j+1 - \left|\left\{(i,j') \mathrel{}\middle|\mathrel{} j' \le j+1 \ \text{and} \ U'_{i,j'} = U_{1,1}\right\}\right| + 1 = \nu(i,j+1). \qedhere
    \end{align*}
\end{proof}

\begin{lemma} \label{lem-boundary-order}
    Suppose $(U'_1,U'_2,<')$ is the result of applying a case II coherent extended Nielsen transformation to $(U_1,U_2,<)$. If $(i,j) \in B'^-$ and $(i,j) <' (i',j')$, then $(i,\nu(i,j)) < (i',\nu(i',j'))$.
\end{lemma}
\begin{proof}
    If $(i',j') \in B'^-$, then this is true by the definition of an extended Nielsen transformation. So suppose $(i',j') \in B'^+$. Then, $(i',j'+1) \in B'^-$, so $(i,\nu(i,j)) < (i',\nu(i',j'+1))$. By Lemma~\ref{lem-nu-j-plus-1}, $\nu(i',j') = \nu(i',j'+1)$, so $(i,\nu(i,j)) < (i',\nu(i',j'))$.
\end{proof}

\begin{lemma} \label{lem-consec-plus}
    Suppose $(U'_1,U'_2,<')$ is the result of applying a case II coherent extended Nielsen transformation to $(U_1,U_2,<)$. If $(i,j) \in B'^+$, then $(i,j+1) \notin B'^+$.
\end{lemma}
\begin{proof}
    Recall that $B'^- = \{(i',\mu(i',j')) \mid (i',j') \in B^-\}$ and $B'^+ = B' \setminus B'^-$. The lemma follows from the observation that $\mu(i',j'+1) \le \mu(i',j') + 2$ for all $i',j'$.
\end{proof}

\begin{lemma} \label{lem-cut-nu}
    Suppose $(U'_1,U'_2,<')$ is the result of applying a case II coherent extended Nielsen transformation to $(U_1,U_2,<)$. If $(i,j)$ cuts $(i',j')$ with respect to $(U'_1,U'_2,<')$, then $(i,\nu(i,j))$ cuts $(i',\nu(i',j'))$ with respect to $(U_1,U_2,<)$.
\end{lemma}
\begin{proof}
    Suppose $(i,j)$ cuts $(i',j')$ with respect to $(U'_1,U'_2,<')$. This means
    \begin{align*}
        (i,j) &\neq (i',j'), \\
        (i,j-1) &<' (i',j'), \quad\text{and} \\
        (i',j'-1) &<' (i,j),
    \end{align*}
    which implies $\{i,i'\} = \{1,2\}$. Thus, $(i,\nu(i,j)) \neq (i',\nu(i',j'))$.

    We claim that $(i,\nu(i,j)-1) < (i',\nu(i',j'))$, which we prove by splitting into a few cases.
    
    If $(i,j-1) \in B'^-$, then $(i,\nu(i,j-1)) < (i',\nu(i',j'))$ by Lemma~\ref{lem-boundary-order}. We have $\nu(i,j) - 1 \le \nu(i,j-1)$, so $(i,\nu(i,j)-1) < (i',\nu(i',j'))$ in this case.
    
    If $(i,j-1) \in B'^+$ and $j \neq 2$, then $(i,j-2) \in B'^-$ by Lemma~\ref{lem-consec-plus}, so $(i,\nu(i,j-2)) < (i',\nu(i',j'))$ by Lemma~\ref{lem-boundary-order}. We have $\nu(i,j-1) - 1 \le \nu(i,j-2)$, and by Lemma~\ref{lem-nu-j-plus-1}, $\nu(i,j-1) = \nu(i,j)$. Hence, $\nu(i,j) - 1 \le \nu(i,j-2)$, so $(i,\nu(i,j)-1) < (i',\nu(i',j'))$ in this case.

    If $(i,j-1) \in B'^+$ and $j=2$, then we claim $i=2$. Indeed, $(1,\mu(1,1)) = (1,1)$, so $(1,1) \in B'^-$. Now, we have $(i,\nu(i,j-1)-1) = (2,1) < (i',\nu(i',j'))$, since $i \not=i'$ and $(2,1)$ is the least element in $B$. By Lemma~\ref{lem-nu-j-plus-1}, $\nu(i,j-1) = \nu(i,j)$, so $(i,\nu(i,j)-1) < (i',\nu(i',j'))$ in this case.

    Finally, if $(i,j-1) \notin B'^- \cup B'^+ = B'$, then $j=1$. If $i=1$, then $\nu(i,j) = 1$, so $(i,\nu(i,j)-1) = (i,0) < (i',\nu(i',j'))$. If $i=2$, then $\nu(i,j) = 2$, so $(i,\nu(i,j)-1) = (2,1) < (i',\nu(i',j'))$.

    We have proved that $(i,\nu(i,j)-1) < (i',\nu(i',j'))$. By symmetry, we can also conclude that $(i',\nu(i',j')-1) < (i,\nu(i,j))$. Therefore, $(i,\nu(i,j))$ cuts $(i',\nu(i',j'))$ with respect to $(U_1,U_2,<)$.
\end{proof}

\begin{lemma} \label{lem-left-inverse}
    Suppose $(U'_1,U'_2,<')$ is the result of applying a case II coherent extended Nielsen transformation to $(U_1,U_2,<)$. If $(i,j) \in B^-$, then $\nu(i,\mu(i,j)) = j$.
\end{lemma}
\begin{proof}
    We claim that
    \[
        |\{(i,j') \mid j' \le j \ \text{and} \ U_{i,j'} = U_{1,1}\}| = \left|\left\{(i,j') \mathrel{}\middle|\mathrel{} j' \le \mu(i,j) \ \text{and} \ U'_{i,j'} = U_{1,1}\right\}\right|.
    \]
    Observe that the left-hand side is the number of variables equal to $U_{1,1}$ in
    \[
        U_{i,1} U_{i,2} \cdots U_{i,j},
    \]
    and the right-hand side is the number of variables equal to $U_{1,1}$ in
    \[
        U'_{i,1} U'_{i,2} \cdots U'_{i,\mu(i,j)} = T(U_{i,1} U_{i,2} \cdots U_{i,j})^+,
    \]
    where $T$ is the endomorphism on $\mathcal{X}^*$ given by $T(U_{1,1}) = U_{2,1} U_{1,1}$ and the identity function on other elements of $\mathcal{X}$. Now, $T$ preserves the number of instances of $U_{1,1}$ in a word, and the first variable in $T(U_{i,1} U_{i,2} \cdots U_{i,j})$ is not $U_{1,1}$. It follows that $U_{i,1} U_{i,2} \cdots U_{i,j}$ and $T(U_{i,1} U_{i,2} \cdots U_{i,j})^+$ have the same number of variables equal to $U_{1,1}$, establishing the claim.

    Now, we have
    \begin{align*}
        \nu(i,\mu(i,j)) &= \mu(i,j)-\left|\left\{(i,j') \mathrel{}\middle|\mathrel{} j' \le \mu(i,j) \ \text{and} \ U'_{i,j'} = U_{1,1}\right\}\right| + 1 \\
        &= j+|\{(i,j') \mid j' \le j \ \text{and} \ U_{i,j'} = U_{1,1}\}|-\left|\left\{(i,j') \mathrel{}\middle|\mathrel{} j' \le \mu(i,j) \ \text{and} \ U'_{i,j'} = U_{1,1}\right\}\right| \\
        &= j. \qedhere
    \end{align*}
\end{proof}

\begin{lemma} \label{lem-nu-minus}
    Suppose $(U'_1,U'_2,<')$ is the result of applying a case II coherent extended Nielsen transformation to $(U_1,U_2,<)$. If $(i,j) \in B'^-$, then $U_{i,\nu(i,j)} = U'_{i,j}$.
\end{lemma}
\begin{proof}
    Let $k$ be such that $\mu(i,k) = j$. By Lemmas~\ref{lem-left-inverse} and \ref{lem-mu-preserve}, we have
    \[
        U_{i,\nu(i,j)} = U_{i,\nu(i,\mu(i,k))} = U_{i,k} = U'_{i,\mu(i,k)} = U'_{i,j}. \qedhere
    \]
\end{proof}

\begin{lemma} \label{lem-nu-plus-1}
    Suppose $(U'_1,U'_2,<')$ is the result of applying a case II coherent extended Nielsen transformation to $(U_1,U_2,<)$. If $(i,j) \in B'^+$, then $(i,\nu(i,j))$ mirrors $(1,1)$ with respect to $(U_1,U_2,<)$.
\end{lemma}
\begin{proof}
    Suppose $(i,j) \in B'^+$. Then, $(i,j+1) \in B'^-$ by Lemma~\ref{lem-consec-plus}, so by Lemmas~\ref{lem-nu-j-plus-1}, \ref{lem-nu-minus}, and \ref{lem-nu-plus-2}, we have
    \[
        U_{i,\nu(i,j)} = U_{i,\nu(i,j+1)} = U'_{i,j+1} = U_{1,1}.
    \]
    It remains to show that $(i,\nu(i,j)) \not\approx (1,1)$. Suppose for contradiction that $(i,\nu(i,j)) \approx (1,1)$. By coherence, $(i,\nu(i,j)-1) \approx (1,0)$, so $\nu(i,j) = 1$. Hence,
    \[
        j = \left|\left\{(i,j') \mathrel{}\middle|\mathrel{} j' \le j \ \text{and} \ U'_{i,j'} = U_{1,1}\right\}\right|,
    \]
    which implies that $U'_{i,j'} = U_{1,1}$ for all $j' \le j$. In particular, $U'_{i,j} = U_{1,1}$, contradicting Lemma~\ref{lem-nu-plus-2}.
\end{proof}

\begin{lemma} \label{lem-nu-neq}
    Suppose $(U'_1,U'_2,<')$ is the result of applying a case II coherent extended Nielsen transformation to $(U_1,U_2,<)$. If $(i,j)$ mirrors $(i',j')$ with respect to $(U'_1,U'_2,<')$ and $U_{i,\nu(i,j)} = U_{i',\nu(i',j')}$, then $(i,\nu(i,j))$ mirrors $(i',\nu(i',j'))$ with respect to $(U_1,U_2,<)$.
\end{lemma}
\begin{proof}
    Suppose $(i,j)$ mirrors $(i',j')$ with respect to $(U'_1,U'_2,<')$ and $U_{i,\nu(i,j)} = U_{i',\nu(i',j')}$. It suffices to show that $(i,\nu(i,j)) \not\approx (i',\nu(i',j'))$. Without loss of generality, $(i,j) <' (i',j')$. If $(i,j) \in B'^-$, then by Lemma~\ref{lem-boundary-order}, $(i,\nu(i,j)) < (i',\nu(i',j'))$, as desired. If $(i,j) \in B'^+$ and $(i',j') \in B'^-$, then by Lemmas~\ref{lem-nu-plus-1}, \ref{lem-nu-minus}, and \ref{lem-nu-plus-2}, we have
    \[
        U_{1,1} = U_{i,\nu(i,j)} = U_{i',\nu(i',j')} = U'_{i',j'} = U'_{i,j} = U_{2,1},
    \]
    a contradiction. Finally, if $(i,j), (i',j') \in B'^+$, then by Lemma~\ref{lem-nu-plus-2}, $U'_{i,j+1} = U'_{i',j'+1} = U_{1,1}$. By coherence, $(i,j+1) <' (i',j'+1)$. Since $(i,j+1), (i',j'+1) \in B'^-$ by Lemma~\ref{lem-consec-plus}, we have $(i,\nu(i,j+1)) < (i',\nu(i',j'+1))$. By Lemma~\ref{lem-nu-j-plus-1}, $\nu(i,j) = \nu(i,j+1)$ and $\nu(i',j') = \nu(i',j'+1)$. Therefore, $(i,\nu(i,j)) < (i',\nu(i',j'))$, as desired.
\end{proof}

\begin{lemma} \label{lem-mirror-case2}
    Suppose $(U'_1,U'_2,<')$ is the result of applying a case II coherent extended Nielsen transformation to $(U_1,U_2,<)$. If $(i,j)$ mirrors $(i',j')$ with respect to $(U'_1,U'_2,<')$ and $(i,\nu(i,j))$ does not mirror $(i',\nu(i',j'))$ with respect to $(U_1,U_2,<)$, then either $(i,j) \in B'^+$ and $(i',j') \in \{(i'',j'') \in B'^- \mid U'_{i'',j''} = U_{2,1}\}$ or else $(i,j) \in \{(i'',j'') \in B'^- \mid U'_{i'',j''} = U_{2,1}\}$ and $(i',j') \in B'^+$.
\end{lemma}
\begin{proof}
    Suppose $(i,j)$ mirrors $(i',j')$ with respect to $(U'_1,U'_2,<')$ and $(i,\nu(i,j))$ does not mirror $(i',\nu(i',j'))$ with respect to $(U_1,U_2,<)$. By Lemma~\ref{lem-nu-neq}, $U_{i,\nu(i,j)} \neq U_{i',\nu(i',j')}$. We first show that either $(i,j) \in B'^+$ and $(i',j') \in B'^-$ or else $(i,j) \in B'^-$ and $(i',j') \in B'^+$. If not, then either $(i,j), (i',j') \in B'^+$ or $(i,j), (i',j') \in B'^-$. In the former case, Lemma~\ref{lem-nu-plus-1} implies that $U_{i,\nu(i,j)} = U_{1,1} = U_{i',\nu(i',j')}$, a contradiction. In the latter case, Lemma~\ref{lem-nu-minus} implies that
    \[
        U_{i,\nu(i,j)} = U'_{i,j} = U'_{i',j'} = U_{i',\nu(i',j')},
    \]
    a contradiction.

    If $(i,j) \in B'^+$ and $(i',j') \in B'^-$, then by Lemma~\ref{lem-nu-plus-2}, $U'_{i',j'} = U'_{i,j} = U_{2,1}$, so $(i',j') \in \{(i'',j'') \in B'^- \mid U'_{i'',j''} = U_{2,1}\}$. Similarly, if $(i,j) \in B'^-$ and $(i',j') \in B'^+$, then $(i,j) \in \{(i'',j'') \in B'^- \mid U'_{i'',j''} = U_{2,1}\}$.
\end{proof}

\begin{lemma} \label{lem-nu-mirror}
    Suppose $(U'_1,U'_2,<')$ is the result of applying a case II coherent extended Nielsen transformation to $(U_1,U_2,<)$. If $(i,j) \in B'^-$ and $U'_{i,j} = U_{2,1}$, then $(i,\nu(i,j))$ mirrors $(2,1)$ with respect to $(U_1,U_2,<)$.
\end{lemma}
\begin{proof}
    Suppose $(i,j) \in B'^-$ and $U'_{i,j} = U_{2,1}$. Then, by Lemma~\ref{lem-nu-minus}, $U_{i,\nu(i,j)} = U'_{i,j} = U_{2,1}$. It remains to show that $(i,\nu(i,j)) \not\approx (2,1)$. If not, then $(i,\nu(i,j)) = (2,1)$, so $i=2$. But $\nu(2,j) \ge 2$ for all $j \ge 1$, a contradiction. Hence, $(i,\nu(i,j))$ mirrors $(2,1)$ with respect to $(U_1,U_2,<)$.
\end{proof}

\begin{lemma} \label{lem-walk-case2}
    Suppose $(U'_1,U'_2,<')$ is the result of applying a case II coherent extended Nielsen transformation to $(U_1,U_2,<)$. If $G_{U_1,U_2,<}$ is acyclic, then so is $G_{U'_1,U'_2,<'}$.
\end{lemma}
\begin{proof}
    It suffices to show that if $((i,j),(i',j')) \in E'$, then there is a walk of length at least 1 from $(i,\nu(i,j))$ to $(i',\nu(i',j'))$ in $G_{U_1, U_2, <}$. So suppose $((i,j),(i',j')) \in E'$. Then, with respect to $(U'_1,U'_2,<')$, $(i,j)$ cuts some $(i'',j'') \in B'$ and $(i'',j'')$ mirrors $(i',j')$. By Lemma~\ref{lem-cut-nu}, $(i,\nu(i,j))$ cuts $(i'',\nu(i'',j''))$ with respect to $(U_1,U_2,<)$. If $(i'',\nu(i'',j''))$ mirrors $(i',\nu(i',j'))$ with respect to $(U_1,U_2,<)$, then $((i,\nu(i,j)),(i',\nu(i',j'))) \in E$, and we are done. Otherwise, by Lemma~\ref{lem-mirror-case2}, either $(i'',j'') \in B'^+$ and $(i',j') \in \{(i''',j''') \in B'^- \mid U'_{i''',j'''} = U_{2,1}\}$ or else $(i'',j'') \in \{(i''',j''') \in B'^- \mid U'_{i''',j'''} = U_{2,1}\}$ and $(i',j') \in B'^+$.

    First, suppose $(i'',j'') \in B'^+$ and $(i',j') \in \{(i''',j''') \in B'^- \mid U'_{i''',j'''} = U_{2,1}\}$. Then, by Lemma~\ref{lem-nu-plus-1}, $(i'',\nu(i'',j''))$ mirrors $(1,1)$ with respect to $(U_1,U_2,<)$, so $((i,\nu(i,j)), (1,1)) \in E$. Also, $(1,1)$ cuts $(2,1)$ with respect to $(U_1,U_2,<)$. By Lemma~\ref{lem-nu-mirror}, $(i',\nu(i',j'))$ mirrors $(2,1)$ with respect to $(U_1,U_2,<)$. Thus, $((1,1), (i',\nu(i',j'))) \in E$. The walk $((i,\nu(i,j)), (1,1), \break (i',\nu(i',j')))$ is our desired walk in $G_{U_1, U_2, <}$.

    Similarly, if $(i'',j'') \in \{(i''',j''') \in B'^- \mid U'_{i''',j'''} = U_{2,1}\}$ and $(i',j') \in B'^+$, then \break $((i,\nu(i,j)), (2,1), (i',\nu(i',j')))$ is our desired walk in $G_{U_1, U_2, <}$.
\end{proof}

\begin{lemma} \label{lem-mu-decrease}
    Suppose $(U'_1,U'_2,<')$ is the result of applying a case II coherent extended Nielsen transformation to $(U_1,U_2,<)$. Then, for every $(i,j) \in B'$, we have $f_{U'_1,U'_2,<'}(i,j) \le f_{U_1,U_2,<}(i,\nu(i,j))$.
\end{lemma}
\begin{proof}
    First, $G_{U'_1,U'_2,<'}$ is acyclic by Lemma~\ref{lem-walk-case2}, so $f_{U'_1,U'_2,<'}$ is well-defined. We may inductively assume that the lemma is true for every $(i',j') \in B'$ such that $((i,j),(i',j')) \in E'$.

    We have
    \[
        f_{U'_1,U'_2,<'}(i,j) = 1 + \sum_{\substack{(i',j') \in B' \\ (i'',j'') \ \text{mirrors} \ (i',j') \\ \text{w.r.t.} \ (U'_1,U'_2,<')}} f_{U'_1,U'_2,<'}(i',j')
    \]
    for some $(i'',j'') \in B'$ such that $(i,j)$ cuts $(i'',j'')$ with respect to $(U'_1,U'_2,<')$. Then, by Lemma~\ref{lem-cut-nu}, $(i,\nu(i,j))$ cuts $(i'',\nu(i'',j''))$ with respect to $(U_1,U_2,<)$.

    First, suppose $(i'',j'') \notin B'^+ \cup \{(i',j') \in B'^- \mid U'_{i',j'} = U_{2,1}\}$. Then, $U'_{i'',j''} \neq U_{2,1}$, so by Lemma~\ref{lem-nu-plus-2}, $(i'',j'')$ does not mirror any boundary in $B'^+$ with respect to $(U'_1,U'_2,<')$. Hence,
    \begin{align*}
        f_{U'_1,U'_2,<'}(i,j) &= 1 + \sum_{\substack{(i',j') \in B' \\ (i'',j'') \ \text{mirrors} \ (i',j') \\ \text{w.r.t.} \ (U'_1,U'_2,<')}} f_{U'_1,U'_2,<'}(i',j') \\
        &= 1 + \sum_{\substack{(i',j') \in B'^- \\ (i'',j'') \ \text{mirrors} \ (i',j') \\ \text{w.r.t.} \ (U'_1,U'_2,<')}} f_{U'_1,U'_2,<'}(i',j') \\
        &\le 1 + \sum_{\substack{(i',j') \in B'^- \\ (i'',j'') \ \text{mirrors} \ (i',j') \\ \text{w.r.t.} \ (U'_1,U'_2,<')}} f_{U_1,U_2,<}(i',\nu(i',j')) \\
        &\le 1 + \sum_{\substack{(i',j') \in B'^- \\ (i'',\nu(i'',j'')) \ \text{mirrors} \ (i',\nu(i',j')) \\ \text{w.r.t.} \ (U_1,U_2,<)}} f_{U_1,U_2,<}(i',\nu(i',j')) \quad \text{by Lemma~\ref{lem-mirror-case2}} \\
        &= 1 + \sum_{\substack{(i',j') \in B^- \\ (i'',\nu(i'',j'')) \ \text{mirrors} \ (i',\nu(i',\mu(i',j'))) \\ \text{w.r.t.} \ (U_1,U_2,<)}} f_{U_1,U_2,<}(i',\nu(i',\mu(i',j'))) \\
        &= 1 + \sum_{\substack{(i',j') \in B^- \\ (i'',\nu(i'',j'')) \ \text{mirrors} \ (i',j') \\ \text{w.r.t.} \ (U_1,U_2,<)}} f_{U_1,U_2,<}(i',j') \quad \text{by Lemma~\ref{lem-left-inverse}} \\
        &\le f_{U_1,U_2,<}(i,\nu(i,j)).
    \end{align*}

    Next, suppose $(i'',j'') \in B'^+$. Then,
    \begin{align*}
        f_{U'_1,U'_2,<'}(i,j) &= 1 + \sum_{\substack{(i',j') \in B' \\ (i'',j'') \ \text{mirrors} \ (i',j') \\ \text{w.r.t.} \ (U'_1,U'_2,<')}} f_{U'_1,U'_2,<'}(i',j') \\
        &\le 1 + \sum_{\substack{(i',j') \in B' \\ (i'',j'') \ \text{mirrors} \ (i',j') \\ \text{w.r.t.} \ (U'_1,U'_2,<')}} f_{U_1,U_2,<}(i',\nu(i',j')) \\
        &= 1 + \smashoperator[r]{\sum_{\substack{(i',j') \in B'^+ \\ (i'',j'') \ \text{mirrors} \ (i',j') \\ \text{w.r.t.} \ (U'_1,U'_2,<')}}} \qquad f_{U_1,U_2,<}(i',\nu(i',j')) + \smashoperator[r]{\sum_{\substack{(i',j') \in B'^- \\ (i'',j'') \ \text{mirrors} \ (i',j') \\ \text{w.r.t.} \ (U'_1,U'_2,<')}}} \qquad f_{U_1,U_2,<}(i',\nu(i',j')).
    \end{align*}
    We bound each of these sums as follows. For the first sum, we have
    \begin{align*}
        \smashoperator[r]{\sum_{\substack{(i',j') \in B'^+ \\ (i'',j'') \ \text{mirrors} \ (i',j') \\ \text{w.r.t.} \ (U'_1,U'_2,<')}}} \qquad f_{U_1,U_2,<}(i',\nu(i',j')) &\le \smashoperator[r]{\sum_{\substack{(i',j') \in B'^+ \\ (i'',\nu(i'',j'')) \ \text{mirrors} \ (i',\nu(i',j')) \\ \text{w.r.t.} \ (U_1,U_2,<)}}} \qquad f_{U_1,U_2,<}(i',\nu(i',j')) \quad \text{by Lemma~\ref{lem-mirror-case2}} \\
        &= \qquad \smashoperator{\sum_{\substack{(i',j') \in B'^+ \\ (i'',\nu(i'',j'')) \ \text{mirrors} \ (i',\nu(i',j'+1)) \\ \text{w.r.t.} \ (U_1,U_2,<)}}} \qquad f_{U_1,U_2,<}(i',\nu(i',j'+1)) \quad \text{by Lemma~\ref{lem-nu-j-plus-1}} \\
        &\le \smashoperator[r]{\sum_{\substack{(i',j') \in B'^- \setminus \{(1,1),(2,1)\} \\ (i'',\nu(i'',j'')) \ \text{mirrors} \ (i',\nu(i',j')) \\ \text{w.r.t.} \ (U_1,U_2,<)}}} \qquad f_{U_1,U_2,<}(i',\nu(i',j')) \quad \text{by Lemma~\ref{lem-consec-plus}} \\
        &\le \sum_{\substack{(i',j') \in B^- \setminus \{(1,1)\} \\ (i'',\nu(i'',j'')) \ \text{mirrors} \ (i',\nu(i',\mu(i',j'))) \\ \text{w.r.t.} \ (U_1,U_2,<)}} f_{U_1,U_2,<}(i',\nu(i',\mu(i',j'))) \\
        &= \sum_{\substack{(i',j') \in B^- \setminus \{(1,1)\} \\ (i'',\nu(i'',j'')) \ \text{mirrors} \ (i',j') \\ \text{w.r.t.} \ (U_1,U_2,<)}} f_{U_1,U_2,<}(i',j') \quad \text{by Lemma~\ref{lem-left-inverse}}.
    \end{align*}
    If $(i',j') \in B'^-$ and $(i'',j'')$ mirrors $(i',j')$ with respect to $(U'_1,U'_2,<')$, then by Lemma~\ref{lem-nu-plus-2}, $U'_{i',j'} = U'_{i'',j''} = U_{2,1}$. Hence, for the second sum, we have
    \begin{align*}
        \smashoperator[r]{\sum_{\substack{(i',j') \in B'^- \\ (i'',j'') \ \text{mirrors} \ (i',j') \\ \text{w.r.t.} \ (U'_1,U'_2,<')}}} \qquad f_{U_1,U_2,<}(i',\nu(i',j')) &\le \sum_{\substack{(i',j') \in B'^- \\ (2,1) \ \text{mirrors} \ (i',\nu(i',j')) \\ \text{w.r.t.} \ (U_1,U_2,<)}} f_{U_1,U_2,<}(i',\nu(i',j')) \quad \text{by Lemma~\ref{lem-nu-mirror}} \\
        &= \sum_{\substack{(i',j') \in B^- \\ (2,1) \ \text{mirrors} \ (i',\nu(i',\mu(i',j'))) \\ \text{w.r.t.} \ (U_1,U_2,<)}} f_{U_1,U_2,<}(i',\nu(i',\mu(i',j'))) \\
        &= \sum_{\substack{(i',j') \in B^- \\ (2,1) \ \text{mirrors} \ (i',j') \\ \text{w.r.t.} \ (U_1,U_2,<)}} f_{U_1,U_2,<}(i',j') \quad \text{by Lemma~\ref{lem-left-inverse}}.
    \end{align*}
    By Lemma~\ref{lem-nu-plus-1}, $(i'',\nu(i'',j''))$ mirrors $(1,1)$ with respect to $(U_1,U_2,<)$. Hence,
    \begin{align*}
        f_{U'_1,U'_2,<'}(i,j) &\le 1 + \sum_{\substack{(i',j') \in B^- \setminus \{(1,1)\} \\ (i'',\nu(i'',j'')) \ \text{mirrors} \ (i',j') \\ \text{w.r.t.} \ (U_1,U_2,<)}} f_{U_1,U_2,<}(i',j') + \sum_{\substack{(i',j') \in B^- \\ (2,1) \ \text{mirrors} \ (i',j') \\ \text{w.r.t.} \ (U_1,U_2,<)}} f_{U_1,U_2,<}(i',j') \\
        &< 1 + \sum_{\substack{(i',j') \in B^- \setminus \{(1,1)\} \\ (i'',\nu(i'',j'')) \ \text{mirrors} \ (i',j') \\ \text{w.r.t.} \ (U_1,U_2,<)}} f_{U_1,U_2,<}(i',j') + f_{U_1,U_2,<}(1,1) \le f_{U_1,U_2,<}(i,\nu(i,j)).
    \end{align*}

    Finally, if $(i'',j'') \in \{(i',j') \in B'^- \mid U'_{i',j'} = U_{2,1}\}$, then we similarly have
    \[
        f_{U'_1,U'_2,<'}(i,j) < 1 + \sum_{\substack{(i',j') \in B^- \\ (i'',\nu(i'',j'')) \ \text{mirrors} \ (i',j') \\ \text{w.r.t.} \ (U_1,U_2,<)}} f_{U_1,U_2,<}(i',j') + f_{U_1,U_2,<}(2,1) \le f_{U_1,U_2,<}(i,\nu(i,j)). \qedhere
    \]
\end{proof}

\begin{lemma} \label{lem-less-21}
    Suppose $(U'_1,U'_2,<')$ is the result of applying a case II coherent extended Nielsen transformation to $(U_1,U_2,<)$. Then,
    \[
        \sum_{(i,j) \in B'^+} f_{U'_1,U'_2,<'}(i,j) < f_{U_1,U_2,<}(2,1).
    \]
\end{lemma}
\begin{proof}
    First, $G_{U'_1,U'_2,<'}$ is acyclic by Lemma~\ref{lem-walk-case2}, so $f_{U'_1,U'_2,<'}$ is well-defined. Now, we have
    \begin{align*}
        \sum_{(i,j) \in B'^+} f_{U'_1,U'_2,<'}(i,j) &\le \sum_{(i,j) \in B'^+} f_{U_1,U_2,<}(i,\nu(i,j)) \quad \text{by Lemma~\ref{lem-mu-decrease}} \\
        &= \sum_{\substack{(i,j) \in B'^+ \\ (1,1) \ \text{mirrors} \ (i,\nu(i,j)) \\ \text{w.r.t.} \ (U_1,U_2,<)}} f_{U_1,U_2,<}(i,\nu(i,j)) \quad \text{by Lemma~\ref{lem-nu-plus-1}} \\
        &= \sum_{\substack{(i,j) \in B'^+ \\ (1,1) \ \text{mirrors} \ (i,\nu(i,j+1)) \\ \text{w.r.t.} \ (U_1,U_2,<)}} f_{U_1,U_2,<}(i,\nu(i,j+1)) \quad \text{by Lemma~\ref{lem-nu-j-plus-1}} \\
        &\le \sum_{\substack{(i,j) \in B'^- \setminus \{(1,1),(2,1)\} \\ (1,1) \ \text{mirrors} \ (i,\nu(i,j)) \\ \text{w.r.t.} \ (U_1,U_2,<)}} f_{U_1,U_2,<}(i,\nu(i,j)) \quad \text{by Lemma~\ref{lem-consec-plus}} \\
        &\le \sum_{\substack{(i,j) \in B^- \setminus \{(1,1)\} \\ (1,1) \ \text{mirrors} \ (i,\nu(i,\mu(i,j))) \\ \text{w.r.t.} \ (U_1,U_2,<)}} f_{U_1,U_2,<}(i,\nu(i,\mu(i,j))) \\
        &= \sum_{\substack{(i,j) \in B^- \setminus \{(1,1)\} \\ (1,1) \ \text{mirrors} \ (i,j) \\ \text{w.r.t.} \ (U_1,U_2,<)}} f_{U_1,U_2,<}(i,j) \quad \text{by Lemma~\ref{lem-left-inverse}} \\
        &< f_{U_1,U_2,<}(2,1). \qedhere
    \end{align*}
\end{proof}

We can now prove Theorem~\ref{thm-sufficient}.

\begin{proof}[Proof of Theorem~\ref{thm-sufficient}]
    If $(U'_1,U'_2,<')$ is the result of applying a case I coherent extended Nielsen transformation to $(U_1,U_2,<)$, then we have
    \begin{align*}
        m(U'_1,U'_2,<') &= \sum_{(i,j) \in B'} f_{U'_1,U'_2,<'}(i,j) \\
        &= \sum_{(i,j) \in B \setminus \{(1,1),(2,1)\}} f_{U'_1,U'_2,<'}(i,j-1) \\
        &\le \sum_{(i,j) \in B \setminus \{(1,1),(2,1)\}} f_{U_1,U_2,<}(i,j) \quad \text{by Lemma~\ref{lem-case1}} \\
        &< \sum_{(i,j) \in B} f_{U_1,U_2,<}(i,j) = m(U_1, U_2, <).
    \end{align*}

    Otherwise, we may assume without loss of generality that $(U'_1,U'_2,<')$ is the result of applying a case II coherent extended Nielsen transformation to $(U_1,U_2,<)$. In this case, we have
    \begin{align*}
        m(U'_1,U'_2,<') &= \sum_{(i,j) \in B'} f_{U'_1,U'_2,<'}(i,j) \\
        &= \sum_{(i,j) \in B'^-} f_{U'_1,U'_2,<'}(i,j) + \sum_{(i,j) \in B'^+} f_{U'_1,U'_2,<'}(i,j) \\
        &= \sum_{(i,j) \in B^-} f_{U'_1,U'_2,<'}(i,\mu(i,j)) + \sum_{(i,j) \in B'^+} f_{U'_1,U'_2,<'}(i,j) \\
        &\le \sum_{(i,j) \in B^-} f_{U_1,U_2,<}(i,\nu(i,\mu(i,j))) + \sum_{(i,j) \in B'^+} f_{U'_1,U'_2,<'}(i,j) \quad \text{by Lemma~\ref{lem-mu-decrease}} \\
        &= \sum_{(i,j) \in B^-} f_{U_1,U_2,<}(i,j) + \sum_{(i,j) \in B'^+} f_{U'_1,U'_2,<'}(i,j) \quad \text{by Lemma~\ref{lem-left-inverse}} \\
        &= \sum_{(i,j) \in B} f_{U_1,U_2,<}(i,j) - f_{U_1,U_2,<}(2,1) + \sum_{(i,j) \in B'^+} f_{U'_1,U'_2,<'}(i,j) \\
        &< \sum_{(i,j) \in B} f_{U_1,U_2,<}(i,j) \quad \text{by Lemma~\ref{lem-less-21}} \\
        &= m(U_1,U_2,<).
    \end{align*}

    It remains to show that $m(U_1,U_2,<) \le 2^N$, where $N = |U_1| + |U_2|$. Consider the function $g_{U_1,U_2,<}$ recursively defined as
    \[
        g_{U_1,U_2,<}(i,j) = 1 + \sum_{\substack{(i',j') \in B \\ ((i,j),(i',j')) \in E}} g_{U_1,U_2,<}(i',j').
    \]
    Then, $f_{U_1,U_2,<}(i,j) \le g_{U_1,U_2,<}(i,j)$, so
    \[
        m(U_1,U_2,<) = \sum_{(i,j) \in B} f_{U_1,U_2,<}(i,j) \le \sum_{(i,j) \in B} g_{U_1,U_2,<}(i,j).
    \]
    But $g_{U_1,U_2,<}(i,j)$ has a graph-theoretic interpretation: it is the number of walks in $G_{U_1, U_2, <}$ starting at $(i,j)$. We therefore conclude by observing that for an acyclic graph on $N$ vertices, the total number of walks is at most $2^N$.
\end{proof}

\section{A necessary condition for termination}

The converse of Theorem~\ref{thm-sufficient} is false, as the counterexample in Figure~\ref{fig-converse} demonstrates.
\begin{figure}[H]
        \centering
        \begin{tikzpicture}
            \draw (0,0) rectangle (3,1);
            \draw (3,0) rectangle (6,1);
            \draw (0,1) rectangle (3,2);
            \draw (3,1) rectangle (6,2);

            \node at (1.5,1.5) {$X$};
            \node at (4.5,1.5) {$X$};
            \node at (1.5,0.5) {$Y$};
            \node at (4.5,0.5) {$Y$};

            \node[scale=0.67] at (3,2.25) {$(1,1)$};
            \node[scale=0.67] at (6,2.25) {$(1,2)$};
            \node[scale=0.67] at (3,-0.25) {$(2,1)$};
            \node[scale=0.67] at (6,-0.25) {$(2,2)$};
        \end{tikzpicture}
        \caption{A terminating extended word equation whose cut graph is cyclic}
        \label{fig-converse}
\end{figure}
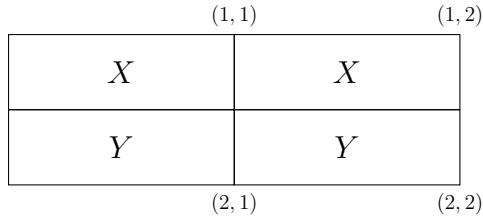

Notice that in the extended word equation shown in Figure~\ref{fig-converse}, we have $(1,1) \approx (2,1)$. Such equivalences at nontrivial (i.e., non-edge) boundaries contribute to the failure of the converse of Theorem~\ref{thm-sufficient}. We use the following definition to rule out nontrivial equivalences of boundaries.

\begin{definition}
    Let $(U_1, U_2, <)$ be an extended word equation, and write
    \begin{align*}
        U_1 &= U_{1,1} U_{1,2} \cdots U_{1,m} \quad\text{and} \\
        U_2 &= U_{2,1} U_{2,2} \cdots U_{2,n},
    \end{align*}
    where $U_{i,j} \in \mathcal{X}$ for all $i,j$. We say that $(U_1, U_2, <)$ is \emph{staggered} if $(1,j) \approx (2,j')$ only for $(j,j') = (0,0)$ and $(j,j') = (m,n)$.
\end{definition}

Even for staggered extended word equations, we sometimes end up in situations where the only coherent extended Nielsen transformations that can be applied to them yield non-staggered extended word equations. For example, in Figure~\ref{fig-converse-2}, if we apply a coherent extended Nielsen transformation, we get a non-staggered extended word equation that soon terminates. This motivates the following definition.

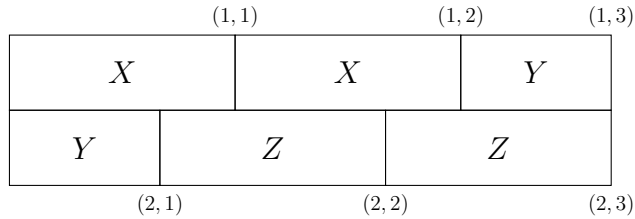
\begin{figure}[H]
        \centering
        \begin{tikzpicture}
            \draw (0,0) rectangle (2,1);
            \draw (2,0) rectangle (5,1);
            \draw (5,0) rectangle (8,1);
            \draw (0,1) rectangle (3,2);
            \draw (3,1) rectangle (6,2);
            \draw (6,1) rectangle (8,2);

            \node at (1.5,1.5) {$X$};
            \node at (4.5,1.5) {$X$};
            \node at (7,1.5) {$Y$};
            \node at (1,0.5) {$Y$};
            \node at (3.5,0.5) {$Z$};
            \node at (6.5,0.5) {$Z$};

            \node[scale=0.67] at (3,2.25) {$(1,1)$};
            \node[scale=0.67] at (6,2.25) {$(1,2)$};
            \node[scale=0.67] at (8,2.25) {$(1,3)$};
            \node[scale=0.67] at (2,-0.25) {$(2,1)$};
            \node[scale=0.67] at (5,-0.25) {$(2,2)$};
            \node[scale=0.67] at (8,-0.25) {$(2,3)$};
        \end{tikzpicture}
        \caption{A terminating staggered extended word equation whose cut graph is cyclic}
        \label{fig-converse-2}
\end{figure}

\begin{definition}
    The set of \emph{hereditarily staggered} extended word equations is the largest set of staggered coherent extended word equations such that
    \begin{itemize}
        \item if $(U_1, U_2, <)$ is hereditarily staggered and nontrivial, then it is possible to apply a coherent extended Nielsen transformation to $(U_1, U_2, <)$ that yields a staggered extended word equation;\footnote{This condition is perhaps less strong than it seems, because it is possible to apply a coherent extended Nielsen transformation to every nontrivial coherent extended word equation (see Proposition~\ref{prop-coherent-ext}).} and
        \item if $(U_1, U_2, <)$ is hereditarily staggered and $(U'_1, U'_2, <')$ is a staggered extended word equation obtained from $(U_1, U_2, <)$ by a coherent extended Nielsen transformation, then $(U'_1, U'_2, <')$ is hereditarily staggered.
    \end{itemize}
\end{definition}

The purpose of this section is to prove the following theorem, which is a partial converse of Theorem~\ref{thm-sufficient}.

\begin{theorem} \label{thm-necessary}
    If $(U_1,U_2,<)$ is hereditarily staggered and terminating, then $G_{U_1,U_2,<}$ is acyclic.
\end{theorem}

We prove Theorem~\ref{thm-necessary} by proving its contrapositive. That is, we start with a hereditarily staggered extended word equation $(U_1,U_2,<)$ such that $G_{U_1,U_2,<}$ is cyclic and prove that there is an infinite sequence of coherent extended Nielsen transformations that can be applied to $(U_1,U_2,<)$. Since $(U_1,U_2,<)$ is chosen arbitrarily and is hereditarily staggered, it suffices to show that it is possible to apply a coherent extended Nielsen transformation to $(U_1,U_2,<)$ that yields a staggered extended word equation $(U'_1,U'_2,<')$ such that $G_{U'_1,U'_2,<'}$ is cyclic. We will show this by proving several lemmas.  First, given a cycle in $G_{U_1,U_2,<}$, Lemmas~\ref{lem-not-minimum} and \ref{lem-minimum} describe how to apply an extended Nielsen transformation to $(U_1,U_2,<)$ that yields a staggered extended word equation $(U'_1,U'_2,<')$ such that $G_{U'_1,U'_2,<'}$ is cyclic. But an issue arises: the extended Nielsen transformations described by Lemmas~\ref{lem-not-minimum} and \ref{lem-minimum} may be incoherent. Fortunately, we have a fallback for that case. Lemma~\ref{lem-fallback} proves that if $(U'_1,U'_2,<')$ is incoherent, then it is possible to apply a coherent extended Nielsen transformation to $(U_1,U_2,<)$ that yields a staggered extended word equation $(U'_1,U'_2,<'')$ such that $G_{U'_1,U'_2,<''}$ is cyclic. This argument reveals the surprising fact that incoherence is related to the presence of cycles in the cut graph (see also Proposition~\ref{prop-acyclic-coherent}).

\subsection{The incoherent case}

We start with a fact from linear algebra, which is related to Farkas' lemma. Given vectors $\mathbf{x}$ and $\mathbf{y}$, the inequality $\mathbf{x} < \mathbf{y}$ should be interpreted coordinate-wise (and similarly with $\le$ in place of $<$). We write $\mathbf{x} \lneq \mathbf{y}$ as an abbreviation for $\mathbf{x} \le \mathbf{y}$ and $\mathbf{x} \neq \mathbf{y}$.
\begin{lemma} \label{lem-alternative-z}
    Let $A$ be an $m \times n$ matrix over $\mathbb{Z}$. Exactly one of the following holds:
    \begin{enumerate}
        \item For some $\mathbf{x} \in \mathbb{Z}^n$ satisfying $\mathbf{x} > \mathbf{0}$, we have $A\mathbf{x} \le \mathbf{0}$.
        \item For some $\mathbf{y} \in \mathbb{Z}^m$ satisfying $y \ge \mathbf{0}$, we have $A^T \mathbf{y} \gneq \mathbf{0}$.
    \end{enumerate}
\end{lemma}
\begin{proof}
    This is stated with $\mathbb{R}$ in place of $\mathbb{Z}$ in \cite[p.~95]{marlow}. Since $\mathbb{Q}$ is dense in $\mathbb{R}$, the result is true over $\mathbb{Q}$ too. Clearing denominators gives the desired result.
\end{proof}

We next prove a simple variant of the previous lemma that is more suited to our purposes.

\begin{lemma} \label{lem-alternative-n}
    Let $A$ and $B$ be $m \times n$ matrices over $\mathbb{N}$. Exactly one of the following holds:
    \begin{enumerate}
        \item For some $\mathbf{x} \in \mathbb{N}^n$ satisfying $\mathbf{x} > \mathbf{0}$, we have $A\mathbf{x} \le B\mathbf{x}$.
        \item For some $\mathbf{y} \in \mathbb{N}^m$, we have $A^T \mathbf{y} \gneq B^T \mathbf{y}$.
    \end{enumerate}
\end{lemma}
\begin{proof}
    Apply Lemma~\ref{lem-alternative-z} to $A-B$.
\end{proof}

Throughout the following discussion, let $(U_1, U_2, <)$ be an extended word equation, and write
 \begin{align*}
    U_1 &= U_{1,1} U_{1,2} \cdots U_{1,m} \quad\text{and} \\
    U_2 &= U_{2,1} U_{2,2} \cdots U_{2,n},
\end{align*}
where $U_{i,j} \in \mathcal{X}$ for all $i,j$. Let
\[
    \ang{i,[j,k]} = \{(i,\ell) \mid \ell \in [j,k]\}.
\]
We think of $\ang{i,[j,k]}$ as representing an interval of boundaries. Note that we may have $k<j$, in which case the interval is empty, although we assume $k \ge j-1$. Let
\[
    U_{\ang{i,[j,k]}} = \sum_{\ell=j}^k U_{i,\ell} \in \mathbb{N}[\mathcal{X}].
\]

\begin{definition}
    Given an extended word equation $(U_1, U_2, <)$, we say that $(\ang{i,[j,k]}, \break \ang{i',[j',k']})$ is a \emph{cover} if
    \begin{align*}
        i &\neq i', \\
        (i',j'-1) &\lesssim (i,j-1), \quad\text{and} \\
        (i,k) &\lesssim (i',k').
    \end{align*}
    If at least one of the above inequalities is strict, we say that $(\ang{i,[j,k]}, \ang{i',[j',k']})$ is a \emph{strict cover}. We say that a cover $(\ang{i,[j,k]}, \ang{i',[j',k']})$ is \emph{tight} if $(\ang{i,[j,k]}, \ang{i',[j'+1,k']})$ and $(\ang{i,[j,k]}, \ang{i',[j',k'-1]})$ are not covers.
\end{definition}
Based on the meaning of the boundary order, if $(\ang{i,[j,k]}, \ang{i',[j',k']})$ is a cover, then $U_{i',j'} U_{i',j'+1} \cdots U_{i',k'}$ is at least as long as $U_{i,j} U_{i,j+1} \cdots U_{i,k}$ in a hypothetical solution; if $(\ang{i,[j,k]}, \ang{i',[j',k']})$ is a strict cover, then $U_{i',j'} U_{i',j'+1} \cdots U_{i',k'}$ is longer than $U_{i,j} U_{i,j+1} \cdots U_{i,k}$ in a hypothetical solution.

Central to our argument is the notion of an \emph{incoherent core}, which isolates the parts of an incoherent extended word equation that are responsible for incoherence.

\begin{definition}
    Given an extended word equation $(U_1, U_2, <)$ and set of covers $C$, let $C' \subseteq C$ be the set of strict covers contained in $C$. An \emph{incoherent core} is a set of covers $C$ such that the following set of inequalities has no solutions over $\mathbb{Z}_{>0}$
    \[
        \{U_{A} \le U_{B} \mid (A, B) \in C\} \cup \{U_{A} < U_{B} \mid (A, B) \in C'\}.
    \]
\end{definition}

Using Lemma~\ref{lem-alternative-n}, we can give a characterization of incoherent cores.

\begin{lemma} \label{lem-incoherent-char}
    Let $C = \{(A_\ell, B_\ell) \mid \ell \in [N]\}$ be a set of covers with respect to some extended word equation $(U_1, U_2, <)$. Then, $C$ is an incoherent core if and only if for some $\{y_\ell \in \mathbb{N} \mid \ell \in [N]\}$, we have
    \[
        \sum_{\ell \in [N]} y_\ell \cdot U_{B_\ell} \subseteq \sum_{\ell \in [N]} y_\ell \cdot U_{A_\ell}
    \]
    such that either the inclusion is strict or $y_\ell > 0$ for some $\ell \in [N]$ such that $(A_\ell, B_\ell)$ is a strict cover.
\end{lemma}
\begin{proof}
    Let $C' \subseteq C$ be the set of strict covers contained in $C$. Since $C$ is an incoherent core, the following set of inequalities has no solutions over $\mathbb{Z}_{>0}$
     \[
        \{U_{A} \le U_{B} \mid (A, B) \in C\} \cup \{U_{A} < U_{B} \mid (A, B) \in C'\}.
    \]
    The lemma follows by adding slack variables to the strict inequalities and appealing to Lemma~\ref{lem-alternative-n}.
\end{proof}

Our goal now is to prove Lemma~\ref{lem-disjoint-incoherent-core}, which says that any incoherent word equation has an incoherent core of a special form that will be useful for forming a cycle. We need a few lemmas on the way to this result.

\begin{lemma} \label{lem-cover-cut}
    Let $(U_1, U_2, <)$ be an extended word equation. If $(A,B)$ is a cover and $(i,j) \in A$ cuts $(i',j')$, then $(i',j') \in B$.
\end{lemma}
\begin{proof}
    Write $B = \ang{i',[\ell,k]}$. We have
    \begin{align*}
        (i',\ell-1) &\lesssim (i,j-1) < (i',j') \quad\text{and} \\
        (i',j'-1) &< (i,j) \lesssim (i',k).
    \end{align*}
    Thus, $j' \in [\ell,k]$, so $(i',j') \in B$.
\end{proof}

\begin{lemma} \label{lem-tight-cut}
    Let $(U_1, U_2, <)$ be an extended word equation. If $(A,B)$ is a tight cover, then any $(i,j) \in B$ cuts some $(i',j') \in A$.
\end{lemma}
\begin{proof}
    Write $(A,B) = (\ang{i',[\ell',k']}, \ang{i,[\ell,k]})$. Since $(A,B)$ is tight,
    \begin{align*}
        (i',\ell'-1) &< (i,\ell) \lesssim (i,j) \quad\text{and} \\
        (i,j-1) &\lesssim (i,k-1) < (i',k').
    \end{align*}
    Thus, there is some $j' \in [\ell',k']$ such that
    \begin{align*}
        (i',j'-1) &< (i,j) \quad\text{and} \\
        (i,j-1) &< (i',j'),
    \end{align*}
    so $(i,j)$ cuts $(i',j') \in A$.
\end{proof}

\begin{definition}
    Let $(A_1,B_1)$ and $(A_2,B_2)$ be covers with respect to some extended word equation. We say that $(A_1,B_1)$ \emph{impinges on} $(A_2,B_2)$ if $B_1 \cap A_2 \neq \emptyset$.
\end{definition}

\begin{lemma} \label{lem-tight}
    Suppose $(A_1,B_1)$ and $(A_2,B_2)$ are covers and $(A_1,B_1)$ is tight with respect to some extended word equation. If $(A_1,B_1)$ impinges on $(A_2,B_2)$, then $(A_2,B_2)$ impinges on $(A_1,B_1)$.
\end{lemma}
\begin{proof}
    Suppose $(A_1,B_1)$ impinges on $(A_2,B_2)$. Then, for some boundary $(i,j)$, we have $(i,j) \in B_1 \cap A_2$. Since $(A_1,B_1)$ is tight, $(i,j)$ cuts some $(i',j') \in A_1$ by Lemma~\ref{lem-tight-cut}. By Lemma~\ref{lem-cover-cut}, $(i',j') \in B_2$. Thus, $(i',j') \in B_2 \cap A_1$, so $(A_2,B_2)$ impinges on $(A_1,B_1)$.
\end{proof}

\begin{lemma} \label{lem-disjoint-incoherent-core}
    If $(U_1, U_2, <)$ is incoherent, then there is an incoherent core such that no cover impinges on another.
\end{lemma}
\begin{proof}
    First, we observe that there is an incoherent core, namely
    \[
        \{(\ang{i, [1,j]}, \ang{i', [1,j']}) \mid (i,j) \lesssim (i',j')\}.
    \]
    
    Now, let $C = \{(A_\ell, B_\ell) \mid \ell \in [N]\}$ be an incoherent core. By Lemma~\ref{lem-incoherent-char}, for some $\{y_\ell \in \mathbb{N} \mid \ell \in [N]\}$, we have
    \begin{equation} \label{subset-1}
        \sum_{\ell \in [N]} y_\ell \cdot U_{B_\ell} \subseteq \sum_{\ell \in [N]} y_\ell \cdot U_{A_\ell}\quad ,
    \end{equation}
    such that either the inclusion is strict or $y_\ell > 0$ for some $\ell \in [N]$ such that $(A_\ell, B_\ell)$ is a strict cover. We may choose $C$ and the coefficients $\{y_\ell \in \mathbb{N} \mid \ell \in [N]\}$ so that $y_\ell > 0$ for all $\ell \in [N]$ and
    \begin{equation} \label{minimizer-1}
        \sum_{\ell \in [N]} y_\ell \cdot (|A_\ell| + |B_\ell|)
    \end{equation}
    is minimal. It follows that all the covers in $C$ are tight.

    We claim that no cover of $C$ impinges on another. For the sake of contradiction, suppose otherwise. Without loss of generality, $B_1 \cap A_2 \neq \emptyset$. By Lemma~\ref{lem-tight}, we have $B_2 \cap A_1 \neq \emptyset$ too. We will contradict our assumption that \eqref{minimizer-1} is minimal.

    Write
    \begin{align*}
        A_1 &= \ang{i, [j_1,k_1]} \\
        B_1 &= \ang{i', [j_2,k_2]} \\
        A_2 &= \ang{i', [j_3,k_3]} \\
        B_2 &= \ang{i, [j_4,k_4]}.
    \end{align*}
    Without loss of generality, $j_1 \le j_4$. Let
    \begin{align*}
        A'_1 &= \ang{i, [j_1,j_4 - 1]} \\
        B'_1 &= \ang{i', [j_2,j_3 - 1]}.
    \end{align*}
    If $k_1 \le k_4$, let
    \begin{align*}
        A'_2 &= \ang{i', [k_2 + 1,k_3]} \\
        B'_2 &= \ang{i, [k_1 + 1,k_4]}.
    \end{align*}
    On the other hand, if $k_1 > k_4$, let
    \begin{align*}
        A'_2 &= \ang{i, [k_4+1,k_1]} \\
        B'_2 &= \ang{i', [k_3+1,k_2]}.
    \end{align*}
    In either case, $(A'_1,B'_1)$ and $(A'_2,B'_2)$ are covers.

    If $k_1 \le k_4$, we have
    \[
        U_{A_1} + U_{A_2} + U_{B'_1} + U_{B'_2} = U_{\ang{i,[j_1,k_4]}} + U_{\ang{i',[j_2,k_3]}} = U_{B_1} + U_{B_2} + U_{A'_1} + U_{A'_2}.
    \]
    Similarly, if $k_1 > k_4$, we have
    \[
        U_{A_1} + U_{A_2} + U_{B'_1} + U_{B'_2} = U_{\ang{i,[j_1,k_1]}} + U_{\ang{i',[j_2,k_2]}} = U_{B_1} + U_{B_2} + U_{A'_1} + U_{A'_2}.
    \]
    Hence,
    \[
        U_{A_1} + U_{A_2} + U_{B'_1} + U_{B'_2} + \sum_{\ell \in [N]} y_\ell \cdot U_{B_\ell} \subseteq U_{B_1} + U_{B_2} + U_{A'_1} + U_{A'_2} + \sum_{\ell \in [N]} y_\ell \cdot U_{A_\ell}.
    \]
    Canceling terms, we get
    \begin{equation} \label{subset-2}
    \begin{split}
        &U_{B'_1} + U_{B'_2} + (y_1 - 1) \cdot U_{B_1} + (y_2 - 1) \cdot U_{B_2} + \sum_{\ell \in [3,N]} y_\ell \cdot U_{B_\ell} \subseteq \\
        &U_{A'_1} + U_{A'_2} + (y_1 - 1) \cdot U_{A_1} + (y_2 - 1) \cdot U_{A_2} + \sum_{\ell \in [3,N]} y_\ell \cdot U_{A_\ell}.
    \end{split}
    \end{equation}
    If the inclusion \eqref{subset-1} is strict, then so is \eqref{subset-2}. If either $(A_1,B_1)$ or $(A_2,B_2)$ is a strict cover, then at least one of $(A'_1,B'_1)$ and $(A'_2,B'_2)$ is a strict cover. Thus, \eqref{subset-2} satisfies the condition from Lemma~\ref{lem-incoherent-char}.

    We claim
    \begin{align*}
        |A'_1| &< |A_1| \\
        |B'_1| &< |B_1|.
    \end{align*}
    Indeed, since $B_2 \cap A_1 \neq \emptyset$, we have $j_4 \le k_1$, so $|A'_1| < |A_1|$. And since $B_1 \cap A_2 \neq \emptyset$, we have $j_3 \le k_2$, so $|B'_1| < |B_1|$. We can similarly show that if $k_1 \le k_4$, then
    \begin{align*}
        |A'_2| &< |A_2| \\
        |B'_2| &< |B_2|,
    \end{align*}
    and if $k_1 > k_4$, then
    \begin{align*}
        |A'_1| + |A'_2| &< |A_1| \\
        |B'_1| + |B'_2| &< |B_1|.
    \end{align*}
    Hence,
    \[
        |A'_1| + |B'_1| + |A'_2| + |B'_2| < |A_1| + |B_1| + |A_2| + |B_2|.
    \]
    This inequality together with the inclusion \eqref{subset-2} contradicts our assumption that \eqref{minimizer-1} is minimal.
\end{proof}

Let $(U_1,U_2,<)$ be a staggered coherent extended word equation, and assume without loss of generality that $(1,1) > (2,1)$. Let $(U'_1,U'_2,<')$ be an extended word equation obtained from $(U_1,U_2,<)$ by an extended Nielsen transformation, and let $B'^+$ be defined as in Section~\ref{sec-ext-nielsen}. If $(i,j) \in B'^+$ with $(i,j)$ adjacent to $(3-i,j')$ with respect to $<'$ for some $j'$, then we have an extended word equation $(U'_1,U'_2,<'')$, where $<''$ is obtained by swapping the order of $(i,j)$ and $(3-i,j')$ in $<'$. We say that $(U'_1,U'_2,<'')$ is a \emph{swap} of $(U'_1,U'_2,<')$.

Given any two staggered extended word equations $(U'_1,U'_2,<')$ and $(U'_1,U'_2,<'')$ obtained from a staggered coherent extended word equation $(U_1,U_2,<)$ by an extended Nielsen transformation, we can reach one from the other by a sequence of swaps. This is because an extended Nielsen transformation is determined by the choice of how to order each boundary $(i,j) \in B'^+$ relative to every boundary of the form $(3-i,j') \in B'$.

\begin{lemma} \label{lem-fallback}
    Let $(U_1,U_2,<)$ be a hereditarily staggered extended word equation. Suppose it is possible to apply an incoherent extended Nielsen transformation to $(U_1,U_2,<)$ that yields a staggered extended word equation $(U'_1,U'_2,<')$. Then, it is possible to apply a coherent extended Nielsen transformation to $(U_1,U_2,<)$ that yields a staggered extended word equation $(U'_1,U'_2,<'')$ such that $G_{U'_1,U'_2,<''}$ is cyclic.
\end{lemma}
\begin{proof}
    By the assumptions of the lemma, $(U_1,U_2,<)$ is nontrivial, so it is possible to apply a coherent extended Nielsen transformation to $(U_1, U_2, <)$ that yields a staggered extended word equation $(U'_1,U'_2,<'')$. For any choice of $(U'_1,U'_2,<')$ and $(U'_1,U'_2,<'')$, we can reach the latter from the former by a sequence of swaps, at least one of which must go from an incoherent extended word equation to a coherent one. Thus, we may choose $(U'_1,U'_2,<')$ and $(U'_1,U'_2,<'')$ so that each is a swap of the other. Let $B'^+$ be defined as in Section~\ref{sec-ext-nielsen}. Then, for some $(i,j) \in B'^+$ and $j'$, the boundary order $<''$ is obtained by swapping the order of adjacent elements $(i,j)$ and $(3-i,j')$ in $<'$. We claim that $G_{U'_1,U'_2,<''}$ is cyclic.

    By Lemma~\ref{lem-disjoint-incoherent-core}, $(U'_1,U'_2,<')$ has an incoherent core $C = \{(A_\ell, B_\ell) \mid \ell \in [N]\}$ such that no cover impinges on another. By Lemma~\ref{lem-incoherent-char}, there are some $\{y_\ell \in \mathbb{N} \mid \ell \in [N]\}$ such that
    \begin{equation} \label{subset-3}
        \sum_{\ell \in [N]} y_\ell \cdot U'_{B_\ell} \subseteq \sum_{\ell \in [N]} y_\ell \cdot U'_{A_\ell}.
    \end{equation}
    By Lemma~\ref{lem-cover-cut}, for any $\ell \in [N]$, if $(i_1,j_1) \in A_\ell$ cuts some $(i'_1,j'_1)$ with respect to $(U'_1,U'_2,<')$, then $(i'_1,j'_1) \in B_\ell$. We first claim that for any $\ell \in [N]$, every $(i_1,j_1) \in A_\ell$ cuts some $(i'_1,j'_1) \in B_\ell$ with respect to $(U'_1,U'_2,<'')$. Since $(U'_1,U'_2,<')$ and $(U'_1,U'_2,<'')$ are swaps of each other, this follows immediately from the previous assertion unless $(i_1,j_1) \in \{(i,j),(i,j+1),(3-i,j'),(3-i,j'+1)\}$. Since $(i,j)$ and $(3-i,j')$ are adjacent with respect to both $<'$ and $<''$, we have
    \begin{align*}
        (i,j-1) &<' (3-i,j') \\
        (3-i,j'-1) &<' (i,j),
    \end{align*}
    and similarly with $<''$ in place of $<'$. Thus, if $(i_1,j_1) = (i,j)$, then $(i_1,j_1)$ cuts $(3-i,j')$ with respect to both $(U'_1,U'_2,<')$ and $(U'_1,U'_2,<'')$. In particular, $(3-i,j') \in B_\ell$ by Lemma~\ref{lem-cover-cut}, which proves the claim when $(i_1,j_1) = (i,j)$. The remaining three cases can be handled similarly.

    Next, we claim that for any $\ell \in [N]$, every $(i'_1,j'_1) \in B_\ell$ mirrors some $(i_2,j_2) \in A_{\ell'}$ with respect to $(U'_1,U'_2,<'')$ for some $\ell' \in [N]$. Let $(i'_1,j'_1) \in B_\ell$. By the inclusion \eqref{subset-3}, there is some $(i_2,j_2) \in A_{\ell'}$ with $U'_{i'_1,j'_1} = U'_{i_2,j_2}$. Since no cover of $C$ impinges on another, $(i'_1,j'_1) \neq (i_2,j_2)$. And since $(U'_1,U'_2,<'')$ is staggered and coherent, $(i'_1,j'_1) \not\approx'' (i_2,j_2)$. Hence, $(i'_1,j'_1)$ mirrors $(i_2,j_2)$ with respect to $(U'_1,U'_2,<'')$.

    Therefore, for any $(i_1,j_1) \in A_\ell$, there is some $(i_2,j_2) \in A_{\ell'}$ such that $((i_1,j_1), (i_2,j_2)) \in E''$, where $G_{U'_1,U'_2,<''} = (B',E'')$. Hence, we have an infinite walk in $G_{U'_1,U'_2,<''}$. It follows that $G_{U'_1,U'_2,<''}$ is cyclic.
\end{proof}

\subsection{The coherent case}

In light of Lemma~\ref{lem-fallback}, we can now ignore coherence issues. Given a hereditarily staggered extended word equation $(U_1,U_2,<)$ such that $G_{U_1,U_2,<}$ is cyclic, we now aim to describe how to apply an extended Nielsen transformation to $(U_1,U_2,<)$ that yields a staggered extended word equation $(U'_1,U'_2,<')$ such that $G_{U'_1,U'_2,<'}$ is cyclic.

\begin{definition}
    Given an extended word equation $(U_1,U_2,<)$ and a cycle
    \[
        C = ((i_1,j_1), (i_2,j_2), \dots, (i_{n-1},j_{n-1}), (i_1,j_1))
    \]
    in $G_{U_1,U_2,<}$, there is a set of boundaries
    \[
        S = \left\{(i_1,j_1), (i'_1,j'_1), (i_2,j_2), (i'_2,j'_2), \dots, (i_{n-1},j_{n-1}), (i'_{n-1},j'_{n-1})\right\}
    \]
    such that for each $k \in [n-1]$, we have $(i_k,j_k)$ cuts $(i'_k,j'_k)$ and $(i'_k,j'_k)$ mirrors $(i_{k+1},j_{k+1})$ (where $(i_n,j_n) = (i_1,j_1)$). We say that $S$ is an \emph{auxiliary set of $C$}. We say that $C$ is \emph{pointed with respect to $S$} if the boundary $(i_1,j_1)$ is a minimal element of $S$ with respect to $<$.
\end{definition}

\begin{lemma} \label{lem-pointed}
    If $G_{U_1,U_2,<}$ is cyclic, then there is some cycle of minimum length in $G_{U_1,U_2,<}$ that is pointed with respect to some auxiliary set.
\end{lemma}
\begin{proof}
    Let
    \[
        C = ((i_1,j_1), (i_2,j_2), \dots, (i_{n-1},j_{n-1}), (i_1,j_1))
    \]
    be a cycle of minimum length in $G_{U_1,U_2,<}$, and let
    \[
        S = \left\{(i_1,j_1), (i'_1,j'_1), (i_2,j_2), (i'_2,j'_2), \dots, (i_{n-1},j_{n-1}), (i'_{n-1},j'_{n-1})\right\}
    \]
    be an auxiliary set of $C$, where for each $k \in [n-1]$, we have $(i_k,j_k)$ cuts $(i'_k,j'_k)$ and $(i'_k,j'_k)$ mirrors $(i_{k+1},j_{k+1})$ (where $(i_n,j_n) = (i_1,j_1)$). If $(i_k,j_k)$ is a minimal element of $S$ with respect to $<$ for some $k \in [n-1]$, then we have the cycle
    \[
        ((i_k,j_k), (i_{k+1},j_{k+1}), \dots, (i_{n-1},j_{n-1}), (i_1,j_1), \dots, (i_{k-1},j_{k-1}), (i_k,j_k)),
    \]
    which is of minimum length in $G_{U_1,U_2,<}$ and pointed with respect to $S$. If $(i'_k,j'_k)$ is a minimal element of $S$ with respect to $<$ for some $k \in [n-1]$, then we have the cycle
    \[
        \left((i'_k,j'_k), (i'_{k-1},j'_{k-1}), \dots, (i'_1,j'_1), (i'_{n-1},j'_{n-1}), \dots, (i'_{k+1},j'_{k+1}), (i'_k,j'_k)\right),
    \]
    which is of minimum length in $G_{U_1,U_2,<}$ and pointed with respect to $S$.
\end{proof}

Throughout the remainder of this section, if $(U_1,U_2,<)$ and $(U'_1,U'_2,<')$ are extended word equations, let their boundaries be $B$ and $B'$ respectively. If $(U'_1,U'_2,<')$ is the result of applying a case II extended Nielsen transformation to $(U_1,U_2,<)$, let $\mu$, $\nu$, $B^-$, $B'^-$, and $B'^+$ be defined as in Section~\ref{sec-ext-nielsen}.

\begin{lemma} \label{lem-mu-order}
    Suppose $(U'_1,U'_2,<')$ is the result of applying a case II coherent extended Nielsen transformation to $(U_1,U_2,<)$. If $(i,j),(i',j') \in B$ and $(i,j) < (i',j')$, then \break $(i,\mu(i,j)) <' (i',\mu(i',j'))$.
\end{lemma}
\begin{proof}
    If $(i,j),(i',j') \in B^-$, then by Lemma~\ref{lem-left-inverse}, $(i,\nu(i,\mu(i,j))) < (i',\nu(i',\mu(i',j')))$, which implies $(i,\mu(i,j)) <' (i',\mu(i',j'))$. Otherwise, $(i,j) = (2,1)$ and $(i',j') \in B^-$, in which case $(i,\mu(i,j)) = (i,0) <' (i',\mu(i',j'))$.
\end{proof}

\begin{lemma} \label{lem-mu-minus-1-order}
    Suppose $(U'_1,U'_2,<')$ is a staggered extended word equation resulting from applying a case II coherent extended Nielsen transformation to $(U_1,U_2,<)$. Let $(i,j), (i',j') \in B^-$, and assume that if $(i,\mu(i,j)-1) \in B'^+$, then $(i,\mu(i,j)-2) \lessdot' (i,\mu(i,j)-1)$. Then, if $(i,j-1) < (i', j')$, we have $(i,\mu(i,j)-1) <' (i',\mu(i',j'))$.
\end{lemma}
\begin{proof}
    Suppose $(i,j-1) < (i', j')$. If $j = 1$, then $(i,j) = (1,1)$, so $(i,\mu(i,j)-1) = (i,0) <' (i',\mu(i',j'))$. Otherwise, $(i,\mu(i,j-1)) <' (i',\mu(i',j'))$ by Lemma~\ref{lem-mu-order}. If $(i,\mu(i,j)-1) \in B'^-$, then $\mu(i,j-1) = \mu(i,j)-1$, so $(i,\mu(i,j)-1) <' (i',\mu(i',j'))$. If $(i,\mu(i,j)-1) \in B'^+$, then $\mu(i,j-1) = \mu(i,j)-2$. Hence, $(i,\mu(i,j-1)) \lessdot' (i,\mu(i,j)-1)$, so $(i,\mu(i,j)-1) <' (i',\mu(i',j'))$ in this case too. Finally, if $(i,\mu(i,j)-1) \notin B'^- \cup B'^+ = B'$, then $\mu(i,j) = 1$, so $(i,\mu(i,j)-1) = (i,0) <' (i',\mu(i',j'))$. In any case, $(i,\mu(i,j)-1) <' (i',\mu(i',j'))$.
\end{proof}

\begin{lemma} \label{lem-mu-mirror}
    Suppose $(U'_1,U'_2,<')$ is the result of applying a case II coherent extended Nielsen transformation to $(U_1,U_2,<)$. If $(i,j), (i',j') \in B^-$ and $(i,j)$ mirrors $(i',j')$ with respect to $(U_1,U_2,<)$, then $(i,\mu(i,j))$ mirrors $(i',\mu(i',j'))$ with respect to $(U'_1,U'_2,<')$.
\end{lemma}
\begin{proof}
    Suppose $(i,j)$ mirrors $(i',j')$ with respect to $(U_1,U_2,<)$. Then, $U_{i,j} = U_{i', j'}$ and $(i,j) \not\approx (i', j')$. By Lemmas~\ref{lem-mu-preserve} and \ref{lem-mu-order}, $U'_{i,\mu(i,j)} = U'_{i',\mu(i',j')}$ and $(i,\mu(i,j)) \not\approx' (i', \mu(i', j'))$, so $(i,\mu(i,j))$ mirrors $(i',\mu(i',j'))$ with respect to $(U'_1,U'_2,<')$.
\end{proof}

\begin{lemma} \label{lem-not-minimum}
    Let $(U_1,U_2,<)$ be a hereditarily staggered extended word equation. Suppose $G_{U_1,U_2,<}$ contains a cycle
    \[
        C = ((i_1,j_1), (i_2,j_2), \dots, (i_{n-1},j_{n-1}), (i_1,j_1))
    \]
    that is pointed with respect to some auxiliary set. If $(i_1,j_1)$ is not the minimum boundary of $B$, then it is possible to apply a coherent extended Nielsen transformation to $(U_1,U_2,<)$ that yields a staggered extended word equation $(U'_1,U'_2,<')$ such that $G_{U'_1,U'_2,<'}$ is cyclic.
\end{lemma}
\begin{proof}
    By Lemma~\ref{lem-fallback}, we may assume that every extended Nielsen transformation applied to $(U_1,U_2,<)$ yielding a staggered extended word equation is in fact a coherent extended Nielsen transformation.

    Without loss of generality, assume that $(1,1) > (2,1)$. Let $(U'_1,U'_2,<')$ be the unique staggered extended word equation resulting from applying a case II extended Nielsen transformation to $(U_1,U_2,<)$ such that for each $(i,j) \in B'^+$, we have $(i,j-1) \lessdot' (i,j)$.

    The cycle $C$ is pointed with respect to some auxiliary set
    \[
        S = \left\{(i_1,j_1), (i'_1,j'_1), (i_2,j_2), (i'_2,j'_2), \dots, (i_{n-1},j_{n-1}), (i'_{n-1},j'_{n-1})\right\},
    \]
    where for each $k \in [n-1]$, we have $(i_k,j_k)$ cuts $(i'_k,j'_k)$ and $(i'_k,j'_k)$ mirrors $(i_{k+1},j_{k+1})$ with respect to $(U_1,U_2,<)$ (where $(i_n,j_n) = (i_1,j_1)$). Since $(i_1,j_1) \in B^-$ and $C$ is pointed with respect to $S$, we have $(i_k,j_k), (i'_k,j'_k) \in B^-$ for each $k \in [n-1]$. It follows that $(i_k,\mu(i_k,j_k)), (i'_k,\mu(i'_k,j'_k)) \in B'$ for each $k \in [n-1]$. We claim that
    \[
        C' = ((i_1,\mu(i_1,j_1)), (i_2,\mu(i_2,j_2)), \dots, (i_{n-1},\mu(i_{n-1},j_{n-1})), (i_1,\mu(i_1,j_1)))
    \]
    is a cycle in $G_{U'_1,U'_2,<'}$.

    First, we show that for each $k \in [n-1]$, we have $(i_k,\mu(i_k,j_k))$ cuts $(i'_k,\mu(i'_k,j'_k))$ with respect to $(U'_1,U'_2,<')$. Since $(i_k,j_k)$ cuts $(i'_k,j'_k)$ with respect to $(U_1,U_2,<)$, we have $(i_k,j_k-1) < (i'_k, j'_k)$ and $(i'_k,j'_k-1) < (i_k,j_k)$. Hence, by Lemma~\ref{lem-mu-minus-1-order}, we have $(i_k,\mu(i_k,j_k)-1) <' (i'_k,\mu(i'_k,j'_k))$ and $(i'_k,\mu(i'_k,j'_k)-1) <' (i_k, \mu(i_k,j_k))$. Therefore, $(i_k,\mu(i_k,j_k))$ cuts $(i'_k,\mu(i'_k,j'_k))$ with respect to $(U'_1,U'_2,<')$.

    Next, for each $k \in [n-1]$, we have $(i'_k,\mu(i'_k,j'_k))$ mirrors $(i_{k+1},\mu(i_{k+1},j_{k+1}))$ with respect to $(U'_1,U'_2,<')$ by Lemma~\ref{lem-mu-mirror}. Therefore, $C'$ is a cycle in $G_{U'_1,U'_2,<'}$.
\end{proof}

\begin{lemma} \label{lem-mu-minus-1-cut}
    Suppose $(U'_1,U'_2,<')$ is a staggered extended word equation resulting from applying a case II coherent extended Nielsen transformation to $(U_1,U_2,<)$. Let $(i,j), (i',j') \in B^-$. Assume that $(i,\mu(i,j)-1) \in B'^+$, $(i,\mu(i,j)-1) \lessdot' (i,\mu(i,j))$, and that if $(i',\mu(i',j')-1) \in B'^+$, then $(i',\mu(i',j')-2) \lessdot' (i',\mu(i',j')-1)$. Then, if $(i,j)$ cuts $(i',j')$ with respect to $(U_1,U_2,<)$, we have $(i,\mu(i,j)-1)$ cuts $(i',\mu(i',j'))$ with respect to $(U'_1,U'_2,<')$.
\end{lemma}
\begin{proof}
    Suppose $(i,j)$ cuts $(i',j')$ with respect to $(U_1,U_2,<)$. Then, $(i,j - 1) < (i',j')$. We have $j \neq 1$, so by Lemma~\ref{lem-mu-order}, $(i,\mu(i,j - 1)) <' (i',\mu(i',j'))$. Since $(i,\mu(i,j)-1) \in B'^+$, we have $\mu(i,j - 1) = \mu(i,j)-2$, so $(i,\mu(i,j)-2) <' (i',\mu(i',j'))$. By Lemma~\ref{lem-mu-minus-1-order}, we have $(i',\mu(i',j') - 1) <' (i,\mu(i,j))$. Together with the assumption that $(i,\mu(i,j)-1) \lessdot' (i,\mu(i,j))$, this implies $(i',\mu(i',j') - 1) <' (i,\mu(i,j)-1)$. Therefore, $(i,\mu(i,j)-1)$ cuts $(i',\mu(i',j'))$ with respect to $(U'_1,U'_2,<')$.
\end{proof}

\begin{lemma} \label{lem-mu-minus-1-mirror}
    Suppose $(U'_1,U'_2,<')$ is a staggered extended word equation resulting from applying a case II coherent extended Nielsen transformation to $(U_1,U_2,<)$. Let $(i,j),(i',j') \in B^-$, and assume that $(i',\mu(i',j')-1) \in B'^+$. Then, if $(i,j)$ mirrors $(2,1)$ with respect to $(U_1,U_2,<)$, we have $(i,\mu(i,j))$ mirrors $(i',\mu(i',j')-1)$ with respect to $(U'_1,U'_2,<')$.
\end{lemma}
\begin{proof}
    Suppose $(i,j)$ mirrors $(2,1)$ with respect to $(U_1,U_2,<)$. Then, $U_{i,j} = U_{2,1}$. By Lemmas~\ref{lem-mu-preserve} and \ref{lem-nu-plus-2},
    \[
        U'_{i,\mu(i, j)} = U_{i,j} = U_{2,1} = U'_{i',\mu(i',j')-1}.
    \]
    Since $(U'_1,U'_2,<')$ is staggered, $(i,\mu(i, j)) \not\approx' (i',\mu(i',j')-1)$. Therefore, $(i,\mu(i, j))$ mirrors $(i',\mu(i',j')-1)$ with respect to $(U'_1,U'_2,<')$.
\end{proof}

\begin{lemma} \label{lem-minimum}
    Let $(U_1,U_2,<)$ be a hereditarily staggered extended word equation. Suppose $G_{U_1,U_2,<}$ contains a cycle
    \[
        C = ((i_1,j_1), (i_2,j_2), \dots, (i_{n-1},j_{n-1}), (i_1,j_1))
    \]
    that is pointed with respect to some auxiliary set. If $(i_1,j_1)$ is the minimum boundary of $B$, then it is possible to apply a coherent extended Nielsen transformation to $(U_1,U_2,<)$ that yields a staggered extended word equation $(U'_1,U'_2,<')$ such that $G_{U'_1,U'_2,<'}$ is cyclic.
\end{lemma}
\begin{proof}
    By Lemma~\ref{lem-fallback}, we may assume that every extended Nielsen transformation applied to $(U_1,U_2,<)$ yielding a staggered extended word equation is in fact a coherent extended Nielsen transformation.

    The cycle $C$ is pointed with respect to some auxiliary set
    \[
        S = \left\{(i_1,j_1), (i'_1,j'_1), (i_2,j_2), (i'_2,j'_2), \dots, (i_{n-1},j_{n-1}), (i'_{n-1},j'_{n-1})\right\},
    \]
    where for each $k \in [n-1]$, we have $(i_k,j_k)$ cuts $(i'_k,j'_k)$ and $(i'_k,j'_k)$ mirrors $(i_{k+1},j_{k+1})$ (where $(i_n,j_n) = (i_1,j_1)$). Without loss of generality, assume that $(1,1) > (2,1)$. Then, $(i_1,j_1) = (2,1)$ and $(i'_1,j'_1) = (1,1)$. Since $(i'_1,j'_1)$ mirrors $(i_2,j_2)$, we have $U_{i_2,j_2} = U_{1,1}$. Thus, in any case II extended Nielsen transformation of $(U_1,U_2,<)$, we have $(i_2,\mu(i_2,j_2)-1) \in B'^+$.

    We may assume that $C$ is a cycle of minimum length in $G_{U_1,U_2,<}$. Indeed, by Lemma~\ref{lem-pointed}, $G_{U_1,U_2,<}$ has a cycle of minimum length that is pointed with respect to some auxiliary set; this cycle either satisfies the conditions of Lemma~\ref{lem-not-minimum}, in which case we are done, or it satisfies the conditions of the current lemma. From this assumption, it follows that $(i_k,j_k) \neq (i_\ell,j_\ell)$ and $(i'_k,j'_k) \neq (i'_\ell,j'_\ell)$ for all $k,\ell \in [n-1]$ such that $k \neq \ell$.

    It also follows that if $U_{i'_k,j'_k} = U_{1,1}$ for some $k \in [n-1]$, then $(i'_k,j'_k) \in \{(1,1), (i_2,j_2)\}$. Indeed, if $U_{i'_k,j'_k} = U_{1,1}$ for some $(i'_k,j'_k) \notin \{(1,1), (i_2,j_2)\}$, then we have the shorter cycle
    \[
        ((i_2,j_2), (i_3,j_3), \dots, (i_k,j_k), (i_2,j_2)).
    \]

    We may also assume that if $(i'_k,j'_k) = (i_2,j_2)$ for some $k \in [n-1]$, then $(i'_2,j'_2) \lesssim (i_k, j_k)$. Indeed, if $(i_k, j_k) < (i'_2,j'_2)$, then we may instead consider the cycle
    \[
        \left((i_1,j_1), (i_2,j_2), (i'_{k-1},j'_{k-1}), (i'_{k-2},j'_{k-2}), \dots, (i'_2,j'_2), (i_{k+1},j_{k+1}), \dots, (i_{n-1},j_{n-1}), (i_1,j_1)\right),
    \]
    in which the roles of $(i_k, j_k)$ and $(i'_2,j'_2)$ are reversed.

    At this point in the proof, we split into cases.

    \textbf{Case 1:} $(i'_k,j'_k) \neq (i_2,j_2)$ for all $k \in [n-1]$. Let $(U'_1,U'_2,<')$ be the unique staggered extended word equation resulting from applying a case II extended Nielsen transformation to $(U_1,U_2,<)$ such that $(i_2,\mu(i_2,j_2)-1) \lessdot' (i_2,\mu(i_2,j_2))$ and $(i,j-1) \lessdot' (i,j)$ for each $(i,j) \in B'^+ \setminus \{(i_2,\mu(i_2,j_2)-1)\}$.

    We claim $(i'_k,j'_k) \neq (2,1)$ for all $k \in [n-1]$. Indeed, if $(i'_k,j'_k) = (2,1)$, then $(i_k,j_k) = (1,1)$. It follows that $(i'_{k-1}, j'_{k-1}) \neq (1,1)$ and $U_{i'_{k-1}, j'_{k-1}} = U_{1,1}$. Hence, $(i'_{k-1}, j'_{k-1}) = (i_2,j_2)$, a contradiction. It follows from this claim that $(i_k,\mu(i_k,j_k)), (i'_k,\mu(i'_k,j'_k)) \in B'$ for each $k \in [2,n-1]$.

    We claim that
    \[
        C' = ((i_2,\mu(i_2,j_2)-1), (i_3,\mu(i_3,j_3)), \dots, (i_{n-1},\mu(i_{n-1},j_{n-1})), (i_2,\mu(i_2,j_2)-1))
    \]
    is a cycle in $G_{U'_1,U'_2,<'}$.

    We have $(i_2,\mu(i_2,j_2)-1)$ cuts $(i'_2,\mu(i'_2,j'_2))$ with respect to $(U'_1,U'_2,<')$ by Lemma~\ref{lem-mu-minus-1-cut}. Also, $(i'_{n-1},\mu(i'_{n-1}, j'_{n-1}))$ mirrors $(i_2,\mu(i_2,j_2)-1)$ with respect to $(U'_1,U'_2,<')$ by Lemma~\ref{lem-mu-minus-1-mirror}. For each $k \in [3,n-1]$, we have $(i_k,\mu(i_k,j_k))$ cuts $(i'_k,\mu(i'_k,j'_k))$ with respect to $(U'_1,U'_2,<')$ by Lemma~\ref{lem-mu-minus-1-order}. For each $k \in [2,n-2]$, we have $(i'_k,\mu(i'_k,j'_k))$ mirrors $(i_{k+1},\mu(i_{k+1},j_{k+1}))$ with respect to $(U'_1,U'_2,<')$ by Lemma~\ref{lem-mu-mirror}. Therefore, $C'$ is a cycle in $G_{U'_1,U'_2,<'}$.

    \textbf{Case 2:} $(i'_k,j'_k) = (2,1)$ for some $k \in [n-1]$. Let $(U'_1,U'_2,<')$ be as in Case 1.

    Let $k \in [n-1]$ be such that $(i'_k,j'_k) = (2,1)$. Then, $(i_k,j_k) = (1,1)$, so $(i'_{k-1}, j'_{k-1}) \neq (1,1)$ and $U_{i'_{k-1}, j'_{k-1}} = U_{1,1}$. Hence, $(i'_{k-1}, j'_{k-1}) = (i_2,j_2)$.

    We claim that
    \begin{align*}
        C' = &((i_2,\mu(i_2,j_2)-1), (i_3,\mu(i_3,j_3)), \dots, (i_{k-1},\mu(i_{k-1},j_{k-1})), \\
        &(i_{k+1},\mu(i_{k+1},j_{k+1})), \dots, (i_{n-1},\mu(i_{n-1},j_{n-1})), (i_2,\mu(i_2,j_2)-1))
    \end{align*}
    is a cycle in $G_{U'_1,U'_2,<'}$.

    By Lemma~\ref{lem-mu-minus-1-cut}, we have $(i_2,\mu(i_2,j_2)-1)$ cuts $(i'_2,\mu(i'_2,j'_2))$ and $(i_{k-1},\mu(i_{k-1},j_{k-1}))$ cuts $(i_2,\mu(i_2,j_2)-1)$, both with respect to $(U'_1,U'_2,<')$. By Lemma~\ref{lem-mu-minus-1-mirror}, $(i_2,\mu(i_2,j_2)-1)$ mirrors $(i_{k+1},\mu(i_{k+1},j_{k+1}))$ and $(i'_{n-1},\mu(i'_{n-1}, j'_{n-1}))$ mirrors $(i_2,\mu(i_2,j_2)-1)$, both with respect to $(U'_1,U'_2,<')$. For each $\ell \in [3,k-2] \cup [k+1,n-1]$, we have $(i_\ell,\mu(i_\ell,j_\ell))$ cuts $(i'_\ell,\mu(i'_\ell,j'_\ell))$ with respect to $(U'_1,U'_2,<')$ by Lemma~\ref{lem-mu-minus-1-order}. For each $\ell \in [2,k-2] \cup [k+1,n-2]$, we have $(i'_\ell,\mu(i'_\ell,j'_\ell))$ mirrors $(i_{\ell+1},\mu(i_{\ell+1},j_{\ell+1}))$ with respect to $(U'_1,U'_2,<')$ by Lemma~\ref{lem-mu-mirror}. Therefore, $C'$ is a cycle in $G_{U'_1,U'_2,<'}$.

    \textbf{Case 3:} $(i'_k,j'_k) = (i_2,j_2)$ for some $k \in [n-1]$ and $(i'_k,j'_k) \neq (2,1)$ for all $k \in [n-1]$. We first show that if $(U'_1,U'_2,<')$ is a case II extended Nielsen transformation of $(U_1,U_2,<)$, then $(i_2,\mu(i_2,j_2)-2) <' (i'_2,\mu(i'_2,j'_2))$ and $(i'_2,\mu(i'_2,j'_2) - 1) <' (i_2,\mu(i_2,j_2))$. Since $(i_2,j_2)$ cuts $(i'_2,j'_2)$ with respect to $(U_1,U_2,<)$, we have $(i_2,j_2 - 1) < (i'_2,j'_2)$ and $(i'_2,j'_2 - 1) < (i_2,j_2)$. We have $j_2,j'_2 \neq 1$ (we have $(i'_2,j'_2) \neq (1,1)$, since $C$ is of minimum length), so by Lemma~\ref{lem-mu-order}, $(i_2,\mu(i_2,j_2 - 1)) <' (i'_2,\mu(i'_2,j'_2))$ and $(i'_2,\mu(i'_2,j'_2 - 1)) <' (i_2,\mu(i_2,j_2))$. Since $(i_2,\mu(i_2,j_2)-1) \in B'^+$, we have $\mu(i_2,j_2 - 1) = \mu(i_2,j_2)-2$. We have $U_{i'_2,j'_2} \neq U_{1,1}$, so $\mu(i'_2,j'_2 - 1) = \mu(i'_2,j'_2) - 1$. Therefore, $(i_2,\mu(i_2,j_2)-2) <' (i'_2,\mu(i'_2,j'_2))$ and $(i'_2,\mu(i'_2,j'_2) - 1) <' (i_2,\mu(i_2,j_2))$.

    In light of the previous paragraph, we can let $(U'_1,U'_2,<')$ be the unique staggered extended word equation resulting from applying a case II extended Nielsen transformation to $(U_1,U_2,<)$ such that $(i'_2,\mu(i'_2,j'_2)-1) <' (i_2,\mu(i_2,j_2) - 1) <' (i'_2,\mu(i'_2,j'_2))$ and $(i,j-1) \lessdot' (i,j)$ for each $(i,j) \in B'^+ \setminus \{(i_2,\mu(i_2,j_2)-1)\}$.

    We claim that
    \[
        C' = ((i_2,\mu(i_2,j_2)-1), (i_3,\mu(i_3,j_3)), \dots, (i_{n-1},\mu(i_{n-1},j_{n-1})), (i_2,\mu(i_2,j_2)-1))
    \]
    is a cycle in $G_{U'_1,U'_2,<'}$.

    First, we show that $(i_2,\mu(i_2,j_2)-1)$ cuts $(i'_2,\mu(i'_2,j'_2))$ with respect to $(U'_1,U'_2,<')$. Above, we showed that $(i_2,\mu(i_2,j_2)-2) <' (i'_2,\mu(i'_2,j'_2))$, and we have $(i'_2,\mu(i'_2,j'_2)-1) <' (i_2,\mu(i_2,j_2) - 1)$ by construction. Therefore, $(i_2,\mu(i_2,j_2)-1)$ cuts $(i'_2,\mu(i'_2,j'_2))$ with respect to $(U'_1,U'_2,<')$.

    Let $k \in [n-1]$ be such that $(i'_k,j'_k) = (i_2,j_2)$. We now show that $(i_k,\mu(i_k,j_k))$ cuts $(i'_k,\mu(i'_k,j'_k))$ with respect to $(U'_1,U'_2,<')$. By Lemma~\ref{lem-mu-minus-1-order}, we have $(i_k,\mu(i_k,j_k) - 1) <' (i'_k,\mu(i'_k,j'_k))$. We have $(i'_2,j'_2) \lesssim (i_k, j_k)$, so by Lemma~\ref{lem-mu-order}, $(i'_2,\mu(i'_2,j'_2)) \lesssim' (i_k, \mu(i_k,j_k))$. Then,
    \[
        (i'_k,\mu(i'_k,j'_k)-1) = (i_2,\mu(i_2,j_2) - 1) <' (i'_2,\mu(i'_2,j'_2)) \lesssim' (i_k, \mu(i_k,j_k))
    \]
    by construction. Therefore, $(i_k,\mu(i_k,j_k))$ cuts $(i'_k,\mu(i'_k,j'_k))$ with respect to $(U'_1,U'_2,<')$.
    
    We have $(i'_{n-1},\mu(i'_{n-1}, j'_{n-1}))$ mirrors $(i_2,\mu(i_2,j_2)-1)$ with respect to $(U'_1,U'_2,<')$ by Lemma~\ref{lem-mu-minus-1-mirror}. For each $\ell \in [3,k-1] \cup [k+1,n-1]$, we have $(i_\ell,\mu(i_\ell,j_\ell))$ cuts $(i'_\ell,\mu(i'_\ell,j'_\ell))$ with respect to $(U'_1,U'_2,<')$ by Lemma~\ref{lem-mu-minus-1-order}. For each $\ell \in [2,n-2]$, we have $(i'_\ell,\mu(i'_\ell,j'_\ell))$ mirrors $(i_{\ell+1},\mu(i_{\ell+1},j_{\ell+1}))$ with respect to $(U'_1,U'_2,<')$ by Lemma~\ref{lem-mu-mirror}. Therefore, $C'$ is a cycle in $G_{U'_1,U'_2,<'}$.
\end{proof}

We can now prove Theorem~\ref{thm-necessary}.

\begin{proof}[Proof of Theorem~\ref{thm-necessary}]
    Let $(U_1,U_2,<)$ be a hereditarily staggered extended word equation such that $G_{U_1,U_2,<}$ is cyclic. By Lemma~\ref{lem-pointed}, $G_{U_1,U_2,<}$ contains a cycle that is pointed with respect to some auxiliary set. By Lemmas~\ref{lem-not-minimum} and \ref{lem-minimum}, it is possible to apply a coherent extended Nielsen transformation to $(U_1,U_2,<)$ that yields a staggered extended word equation $(U'_1,U'_2,<')$ such that $G_{U'_1,U'_2,<'}$ is cyclic. Since $(U_1,U_2,<)$ is hereditarily staggered, so is $(U'_1,U'_2,<')$. This can be iterated \emph{ad infinitum}, which proves that $(U_1,U_2,<)$ is non-terminating.
\end{proof}

\section{Miscellaneous propositions} \label{sec-misc}

This section contains miscellaneous propositions about extended word equations.

\subsection{Applying a coherent extended Nielsen transformation}

\begin{proposition} \label{prop-coherent-ext}
    If $(U_1, U_2, <)$ is a nontrivial coherent extended word equation, there is some coherent extended Nielsen transformation that can be applied to it.
\end{proposition}
\begin{proof}
    Let $L : \mathcal{X}^* \to (\mathbb{N}, +)$ be a homomorphism such that $L(X) > 0$ for all $X \in \mathcal{X}$ and
    \[
        (i,j) < (i',j') \Longleftrightarrow L(U_{i,1} U_{i,2} \cdots U_{i,j}) < L(U_{i',1} U_{i',2} \cdots U_{i',j'})
    \]
    for all $(i,j), (i',j') \in B$, where $B$ is the set of boundaries of $(U_1,U_2)$.
    
    If $(1,1) \approx (2,1)$, then let $(U'_1, U'_2, <')$ be the case I extended Nielsen transformation to $(U_1, U_2, <)$. In this case, the same homomorphism $L$ demonstrates that $(U'_1, U'_2, <')$ is coherent.
    
    If $(1,1) \not\approx (2,1)$, we may assume without loss of generality that $(1,1) > (2,1)$, since cases II and III are symmetric. Let $L' : \mathcal{X}^* \to (\mathbb{N}, +)$ be the homomorphism given by $L'(U_{1,1}) = L(U_{1,1}) - L(U_{2,1})$ and $L'(X) = L(X)$ for all $X \in \mathcal{X} \setminus \{U_{1,1}\}$. Then, let $(U'_1, U'_2)$ be the case II Nielsen transformation of $(U_1, U_2)$, and let $<'$ be the boundary order given by
    \[
        (i,j) <' (i',j') \Longleftrightarrow L'(U'_{i,1} U'_{i,2} \cdots U'_{i,j}) < L'(U'_{i',1} U'_{i',2} \cdots U'_{i',j'})
    \]
    for all $(i,j), (i',j') \in B'$, where $B'$ is the set of boundaries of $(U'_1,U'_2)$. Let $\mu$, $\nu$, and $B'^-$ be defined as in Section~\ref{sec-ext-nielsen}. Given $(i,j) \in B'^-$, let $k$ be such that $\mu(i,k) = j$. Then, we have
    \begin{align*}
        L'(U'_{i,1} U'_{i,2} \cdots U'_{i,j}) &= L'(U'_{i,1} U'_{i,2} \cdots U'_{i,\mu(i,k)}) \\
        &= L'(T(U_{i,1} U_{i,2} \cdots U_{i,k})^+) \\
        &= L'(T(U_{i,1} U_{i,2} \cdots U_{i,\nu(i,\mu(i,k))})^+) \quad\text{by Lemma~\ref{lem-left-inverse}} \\
        &= L'(T(U_{i,1} U_{i,2} \cdots U_{i,\nu(i,j)})^+) \\
        &= L'(T(U_{i,1} U_{i,2} \cdots U_{i,\nu(i,j)})) - L'(U_{2,1}) \\
        &= L(U_{i,1} U_{i,2} \cdots U_{i,\nu(i,j)}) - L'(U_{2,1}),
    \end{align*}
    where $T$ is the endomorphism on $\mathcal{X}^*$ given by $T(U_{1,1}) = U_{2,1} U_{1,1}$ and the identity function on other elements of $\mathcal{X}$. Hence, for any $(i,j), (i',j') \in B'^-$, we have
    \begin{align*}
        (i,j) <' (i',j') &\Longleftrightarrow L'(U'_{i,1} U'_{i,2} \cdots U'_{i,j}) < L'(U'_{i',1} U'_{i',2} \cdots U'_{i',j'}) \\
        &\Longleftrightarrow L(U_{i,1} U_{i,2} \cdots U_{i,\nu(i,j)}) < L(U_{i',1} U_{i',2} \cdots U_{i',\nu(i',j')}) \\
        &\Longleftrightarrow (i,\nu(i,j)) < (i',\nu(i',j')).
    \end{align*}
    Therefore, $(U'_1, U'_2, <')$ is a case II coherent extended Nielsen transformation of $(U_1, U_2, <)$.
\end{proof}

\subsection{Preserving acyclicity}

\begin{proposition}
    Suppose $(U'_1,U'_2,<')$ is the result of applying a coherent extended Nielsen transformation to $(U_1,U_2,<)$. If $G_{U_1,U_2,<}$ is acyclic, then so is $G_{U'_1,U'_2,<'}$.
\end{proposition}
\begin{proof}
    This is immediate from Lemmas~\ref{lem-walk-case1} and \ref{lem-walk-case2}.
\end{proof}

\subsection{A sufficient condition for coherence}

In this subsection, we prove the following proposition.
\begin{proposition} \label{prop-acyclic-coherent}
    Let $(U_1,U_2,<)$ be an extended word equation. If $G_{U_1,U_2,<}$ is acyclic, then $(U_1,U_2,<)$ is coherent.
\end{proposition}

\noindent
First, we need a definition and a lemma.

\begin{definition}
    Given an extended word equation $(U_1,U_2,<)$, we say that the boundaries $(i,j)$ and $(i',j')$ are \emph{twins} if $U_{i,j} = U_{i',j'}$, $(i,j) \neq (i',j')$, $(i,j) \approx (i',j')$, and $(i,j-1) \approx (i',j'-1)$. We say that a cover $(A,B)$ is \emph{twin-free} if neither $A$ nor $B$ contains a twin.
\end{definition}

\begin{lemma} \label{lem-twin-free-incoherent-core}
    If $(U_1,U_2,<)$ is incoherent, then there is an incoherent core such that no cover impinges on another and every cover is twin-free.
\end{lemma}
\begin{proof}
    Let $C = \{(A_\ell, B_\ell) \mid \ell \in [N]\}$ be an incoherent core. By Lemma~\ref{lem-incoherent-char}, for some $\{y_\ell \in \mathbb{N} \mid \ell \in [N]\}$, we have
    \begin{equation} \label{subset-4}
        \sum_{\ell \in [N]} y_\ell \cdot U_{B_\ell} \subseteq \sum_{\ell \in [N]} y_\ell \cdot U_{A_\ell}
    \end{equation}
    such that either the inclusion is strict or $y_\ell > 0$ for some $\ell \in [N]$ such that $(A_\ell, B_\ell)$ is a strict cover. We may choose $C$ and the coefficients $\{y_\ell \in \mathbb{N} \mid \ell \in [N]\}$ so that $y_\ell > 0$ for all $\ell \in [N]$ and
    \begin{equation} \label{minimizer-2}
        \sum_{\ell \in [N]} y_\ell \cdot (|A_\ell| + |B_\ell|)
    \end{equation}
    is minimal. As in Lemma~\ref{lem-disjoint-incoherent-core}, all the covers in $C$ are tight and no cover impinges on another.

    We claim that every cover is twin-free. For the sake of contradiction, suppose otherwise. Without loss of generality, $(A_1,B_1)$ is not twin-free. Since $(A_1,B_1)$ is tight, both $A_1$ and $B_1$ contain a twin. Let $(i,j)$ and $(i',j')$ be twins such that $(i,j) \in A_1$ and $(i',j') \in B_1$. We will contradict our assumption that \eqref{minimizer-2} is minimal.
    
    Write
    \begin{align*}
        A_1 &= \ang{i, [j_1,k_1]} \\
        B_1 &= \ang{i', [j_2,k_2]}.
    \end{align*}
    Let
    \begin{align*}
        A'_1 &= \ang{i, [j_1,j-1]} \\
        B'_1 &= \ang{i', [j_2,j'-1]} \\
        A''_1 &= \ang{i, [j+1,k_1]} \\
        B''_1 &= \ang{i', [j'+1,k_2]}.
    \end{align*}
    Both $(A'_1,B'_1)$ and $(A''_1,B''_1)$ are covers. We have
    \[
        U_{A_1} + U_{B'_1} + U_{B''_1} = U_{A_1} + U_{B_1} - U_{i',j'} = U_{A_1} + U_{B_1} - U_{i,j} = U_{B_1} + U_{A'_1} + U_{A''_1}.
    \]
    Hence,
    \[
        U_{A_1} + U_{B'_1} + U_{B''_1} + \sum_{\ell \in [N]} y_\ell \cdot U_{B_\ell} \subseteq U_{B_1} + U_{A'_1} + U_{A''_1} + \sum_{\ell \in [N]} y_\ell \cdot U_{A_\ell}.
    \]
    Canceling terms, we get
    \begin{equation} \label{subset-5}
        U_{B'_1} + U_{B''_1} + (y_1 - 1) \cdot U_{B_1} + \sum_{\ell \in [2,N]} y_\ell \cdot U_{B_\ell} \subseteq U_{A'_1} + U_{A''_1} + (y_1 - 1) \cdot U_{A_1} + \sum_{\ell \in [2,N]} y_\ell \cdot U_{A_\ell}.
    \end{equation}
    If the inclusion \eqref{subset-4} is strict, then so is \eqref{subset-5}. If $(A_1,B_1)$ is a strict cover, then at least one of $(A'_1,B'_1)$ and $(A''_1,B''_1)$ is a strict cover. Thus, \eqref{subset-5} satisfies the condition from Lemma~\ref{lem-incoherent-char}.

    We have
    \[
        |A'_1| + |B'_1| + |A''_1| + |B''_1| < |A_1| + |B_1|.
    \]
    This inequality together with the inclusion \eqref{subset-5} contradicts our assumption that \eqref{minimizer-2} is minimal.
\end{proof}

We can now prove Proposition~\ref{prop-acyclic-coherent}.
\begin{proof}[Proof of Proposition~\ref{prop-acyclic-coherent}]
    We prove the contrapositive. Suppose that $(U_1,U_2,<)$ is incoherent. By Lemma~\ref{lem-twin-free-incoherent-core}, $(U_1,U_2,<)$ has an incoherent core $C = \{(A_\ell, B_\ell) \mid \ell \in [N]\}$ such that no cover impinges on another and every cover is twin-free.
    
    By Lemma~\ref{lem-incoherent-char}, there are some $\{y_\ell \in \mathbb{N} \mid \ell \in [N]\}$ such that
    \begin{equation} \label{subset-6}
        \sum_{\ell \in [N]} y_\ell \cdot U_{B_\ell} \subseteq \sum_{\ell \in [N]} y_\ell \cdot U_{A_\ell}.
    \end{equation}
    By Lemma~\ref{lem-cover-cut}, for any $\ell \in [N]$, if $(i_1,j_1) \in A_\ell$ cuts some $(i'_1,j'_1)$, then $(i'_1,j'_1) \in B_\ell$. In particular, for any $\ell \in [N]$, every $(i_1,j_1) \in A_\ell$ cuts some $(i'_1,j'_1) \in B_\ell$.

    Next, we claim that for any $\ell \in [N]$, every $(i'_1,j'_1) \in B_\ell$ mirrors some $(i_2,j_2) \in A_{\ell'}$ for some $\ell' \in [N]$. Let $(i'_1,j'_1) \in B_\ell$. By the inclusion \eqref{subset-6}, there is some $(i_2,j_2) \in A_{\ell'}$ with $U_{i'_1,j'_1} = U_{i_2,j_2}$. Since no cover of $C$ impinges on another, $(i'_1,j'_1) \neq (i_2,j_2)$. Since every cover of $C$ is twin-free, $(i'_1,j'_1) \not\approx (i_2,j_2)$ or $(i'_1,j'_1 - 1) \not\approx (i_2,j_2 - 1)$. Hence, $(i'_1,j'_1)$ mirrors $(i_2,j_2)$.

    Therefore, for any $(i_1,j_1) \in A_\ell$, there is some $(i_2,j_2) \in A_{\ell'}$ such that $((i_1,j_1), (i_2,j_2)) \in E$, where $G_{U_1,U_2,<} = (B,E)$. Hence, we have an infinite walk in $G_{U_1,U_2,<}$. It follows that $G_{U_1,U_2,<}$ is cyclic.
\end{proof}

\section{Future directions}

The most obvious future direction is to find a necessary and sufficient condition for a coherent extended word equation to be terminating. This will require a better understanding of non-staggered extended word equations. It is somewhat vague what would count as a satisfactory necessary and sufficient condition, but we at least expect the following to be true.
\begin{conjecture}
    It is decidable whether a coherent extended word equation is terminating.
\end{conjecture}

It would also be interesting to prove results about which coherent extended word equations are weakly terminating. We expect that the methods developed in this paper can help make progress on that front too.

\section*{Acknowledgments}
This work was funded in part by the Stanford Center for Automated Reasoning (Centaur) and by the Stanford CURIS program for undergraduate research. We are grateful to Joel Day, Harun Khan, Ying Sheng, and Yoni Zohar for helpful discussions.

\bibliography{bib}
\bibliographystyle{amsplain}

\end{document}